\PassOptionsToPackage{colorlinks,linkcolor={blue},citecolor={blue},urlcolor={red},breaklinks=true,final}{hyperref}

\documentclass[nolimits,a4paper,final]{llncs}

\usepackage[utf8]{inputenc}

\usepackage{ifdraft}

\ifdraft
{
  \usepackage[layout=footnote,draft]{fixme}

  \usepackage[notcite,notref]{showkeys}
  \renewcommand*\showkeyslabelformat[1]{%
    {\normalfont\tiny\ttfamily\kern-3ex #1\kern-3ex}
  }
}{
  \usepackage[layout=footnote,final]{fixme}
}

\usepackage{etex}
\usepackage{amssymb,stmaryrd}
\usepackage[nosumlimits]{amsmath}
\usepackage{xcolor}
\usepackage{stackengine}
\usepackage{needspace}
\usepackage{graphbox}
\usepackage{enumitem}
\usepackage{makeidx}
\usepackage{proof}
\usepackage{caption}
\usepackage{subcaption}
\usepackage{wrapfig}
\usepackage{mathtools}
\usepackage{wasysym}
\usepackage{textcomp}
\usepackage{microtype}

\usepackage{tikz}
\usepackage{pgfplots}

\usetikzlibrary{arrows}
\usetikzlibrary{patterns}
\usetikzlibrary{calc}
\usetikzlibrary{fit,decorations.pathreplacing,arrows.meta,positioning}

\tikzset{>=stealth}

\tikzset{
    tension/.initial=1,
    xwobble/.initial=.5,
    sigma line/.style={
        xshift/.initial=10,
        to path={
          .. controls ($ (\tikztostart  -| \tikztotarget) !\pgfkeysvalueof{/tikz/tension} * \pgfkeysvalueof{/tikz/xwobble}! (\tikztostart) $) and
                     ($ (\tikztotarget -| \tikztostart)  !\pgfkeysvalueof{/tikz/tension} * (1-\pgfkeysvalueof{/tikz/xwobble})! (\tikztotarget)  $)
          .. (\tikztotarget)\tikztonodes
        }
    },
    ystretch/.initial=20,
    sigma line/.default=.5cm
}

\usepackage{tikz-cd}

\tikzset{
    commutative diagrams/.cd,
    arrow style=tikz,
    diagrams={>=stealth},
    row sep=large,
    column sep = huge
}

\usepackage{pigpen}
\newcommand{\pbk}{\arrow[dr, phantom, "\text{\tiny\pigpenfont A}", pos=0.05]}

\providecommand{\catname}{\mathbf} 
\providecommand{\clsname}{\mathcal}
\providecommand{\oname}[1]{\operatorname{\mathsf{#1}}}

\def\defcatname#1{\expandafter\def\csname B#1\endcsname{\catname{#1}}}
\def\defcatnames#1{\ifx#1\defcatnames\else\defcatname#1\expandafter\defcatnames\fi}
\defcatnames ABCDEFGHIJKLMNOPQRSTUVWXYZ\defcatnames

\def\defclsname#1{\expandafter\def\csname C#1\endcsname{\clsname{#1}}}
\def\defclsnames#1{\ifx#1\defclsnames\else\defclsname#1\expandafter\defclsnames\fi}
\defclsnames ABCDEFGHIJKLMNOPQRSTUVWXYZ\defclsnames

\def\defbbname#1{\expandafter\def\csname BB#1\endcsname{\mathbb{#1}}}
\def\defbbnames#1{\ifx#1\defbbnames\else\defbbname#1\expandafter\defbbnames\fi}
\defbbnames ABCDEFGHIJKLMNOPQRSTUVWXYZ\defbbnames

\def\Set{\catname{Set}}

\providecommand{\eps}{\operatorname\epsilon}

\providecommand{\from}{\leftarrow}

\providecommand{\argument}{\operatorname{-\!-}}

\providecommand{\ul}{\underline}			                     %

\DeclareOldFontCommand{\bf}{\normalfont\bfseries}{\mathbf}
\providecommand{\PSet}{{\mathcal P}}			                 %

\providecommand{\id}{\mathsf{id}}
\providecommand{\op}{\mathsf{op}}
\providecommand{\comp}{\mathbin{\circ}}
\providecommand{\iso}{\mathbin{\cong}}

\providecommand{\bang}{\operatorname!}				             %

\providecommand{\ito}{\hookrightarrow}				             %
\providecommand{\tto}{\mathrel{\Rightarrow}}			         %
\providecommand{\mto}{\mapsto}
\providecommand{\xto}[1]{\,\xrightarrow{#1}\,}
\providecommand{\xfrom}[1]{\,\xleftarrow{\;#1}\,}
\providecommand{\To}{\mathrel{\Rightarrow}}			           %

\providecommand{\dar}{\kern-1.2pt\operatorname{\downarrow}}	
\providecommand{\uar}{\kern-1.2pt\operatorname{\uparrow}}

\providecommand{\bigor}{\bigvee}
\providecommand{\bigand}{\bigwedge}

\providecommand{\impl}{\Rightarrow}

\providecommand{\fst}{\oname{fst}}
\providecommand{\snd}{\oname{snd}}
\providecommand{\pr}{\oname{pr}}

\providecommand{\brks}[1]{\langle #1\rangle}
\providecommand{\Brks}[1]{\bigl\langle #1\bigr\rangle}

\providecommand{\inl}{\oname{inl}}
\providecommand{\inr}{\oname{inr}}
\providecommand{\inj}{\oname{in}}

\DeclareSymbolFont{Symbols}{OMS}{cmsy}{m}{n}
\DeclareMathSymbol{\iobj}{\mathord}{Symbols}{"3B}
\providecommand{\curry}{\oname{curry}}

\providecommand{\ev}{\oname{ev}}

\usepackage{stmaryrd}

\providecommand{\lsem}{\llbracket}
\providecommand{\rsem}{\rrbracket}
\providecommand{\sem}[1]{\lsem #1 \rsem}

\providecommand{\comma}{,\operatorname{}\linebreak[1]}		 %
\providecommand{\erule}{\rule{0pt}{0pt}}		               %

\providecommand{\by}[1]{\text{/\!\!/~#1}}			             %
\providecommand{\pacman}[1]{}					                     %

\newcommand{\undefine}[1]{\let #1\relax}					                       %

\providecommand{\mone}{{\text{\kern.5pt\rmfamily-}\mathsf{\kern-.5pt1}}}

\providecommand{\smin}{\smallsetminus}

\makeatletter
\@ifpackageloaded{enumitem}{}{
  \usepackage[loadonly]{enumitem}                          %
}                                                          %
\makeatother

\newlist{citemize}{itemize}{1}
\setlist[citemize]{label=\labelitemi,wide} 

\newlist{cenumerate}{enumerate}{1}
\setlist[cenumerate,1]{label=\arabic*.~,ref={\arabic*},wide} 
\makeatletter
\def\mfix#1{\oname{#1}\@ifnextchar\bgroup\@mfix{}}	       %
\def\@mfix#1{#1\@ifnextchar\bgroup\mfix{}}			           %
\makeatother

\providecommand{\case}[3]{\mfix{case}{\mathbin{}#1}{of}{#2}{\kern-1pt;}{\mathbin{}#3}}

\newcommand\sep{\mathrel{\ooalign{\kern1pt\rule[-1pt]{.41pt}{7pt}\kern1pt}}}

\FXRegisterAuthor{sg}{put}{SG}	%
\FXRegisterAuthor{mp}{amp}{MP}	%
\FXRegisterAuthor{ok}{aok}{OK}	%

\makeatletter
\newcommand{\superimpose}[2]{%
  {\ooalign{$#1\@firstoftwo#2$\cr\hfil$#1\@secondoftwo#2$\hfil\cr}}}
\makeatother

\usepackage{marvosym}

\usepackage{savesym}

\savesymbol{degree}
\savesymbol{leftmoon}
\savesymbol{rightmoon}
\savesymbol{fullmoon}
\savesymbol{newmoon}
\savesymbol{diameter}
\savesymbol{emptyset}
\savesymbol{bigtimes}
\savesymbol{triangleright}
\savesymbol{not}

\usepackage[matha]{mathabx}
\restoresymbol{other}{not}
\restoresymbol{other}{emptyset}
\restoresymbol{other}{triangleright}

\newcommand{\dom}{\oname{dom}}
\newcommand{\img}{\oname{img}}

\usepackage{bbold}

\newcommand{\nat}{\mathbb{N}}

\newcommand{\Hpl}{H}
\newcommand{\Hplb}{\CH}

\newcommand{\Hid}{P}
\newcommand{\St}{T}

\newcommand{\TA}{\Theta} %

\let\Ref\undefined
\newcommand{\Ref}{\oname{Ref}}
\newcommand{\CType}{\oname{CType}}

\newcommand{\inacc}{{\scalebox{.8}{$\boxtimes$}}}

\newcommand{\Nat}{\oname{Nat}}

\newcommand{\mbind}[2]{\mfix{let\,}{#1}{in}{#2}}
\newcommand{\letref}[2]{\mfix{letref}{#1}{in}{#2}}
\newcommand{\ret}{\oname{ret}}

\newcommand{\set}[1]{\{#1\}}
\newcommand{\htri}[3]{\set{#1}{#2}\set{#3}}

\renewcommand{\paragraph}[1]{{\setlength{\parskip}{1.5ex}\par\noindent\textbf{#1}~~}}

\newcommand{\ttrue}{\oname{true}}
\newcommand{\ffalse}{\oname{false}}
\newcommand{\nnew}{\oname{new}}
\newcommand{\ass}{\mathrel{:\kern-2pt\text{\textdblhyphen}}}

\newcommand{\ctx}{\vdash}
\newcommand{\cctx}{\ctx_{\mathsf{c}}}
\newcommand{\vctx}{\ctx_{\mathsf{v}}}

\newcommand{\bigior}{\bigcurlyvee}
\newcommand{\bigiand}{\bigcurlywedge}

\newcommand{\pred}{\mathsf{P}}

\DeclareRobustCommand{\ite}[3]{\mfix{{if}\,}{#1}{\,{then}\,}{#2}{\,{else}\,}{#3}}

\renewcommand{\xto}[1]{
\newdimen\stringwidth
\setbox0=\hbox{\tiny$#1$}
\stringwidth=\wd0
  \mathrel{\raisebox{-.1ex}{\kern3pt\ensuremath{\mathrel{\tikz{\draw [-stealth,line width=0.4] (0.6ex,.1ex) -- node[font=\tiny, midway,above=-.3ex,xshift=-.2ex] {$#1$\kern2pt\erule} ++([xshift=2ex]\stringwidth,0) ;}
  }}\kern3pt}}
}

\renewcommand{\xfrom}[1]{
\newdimen\stringwidth
\setbox0=\hbox{\tiny$#1$}
\stringwidth=\wd0
  \mathrel{\raisebox{-.1ex}{\kern3pt\ensuremath{\mathrel{\tikz{\draw [stealth-,line width=0.4] (0.6ex,.1ex) -- node[font=\tiny, midway,above=-.3ex,xshift=.2ex] {$#1$\kern2pt\erule} ++([xshift=2ex]\stringwidth,0) ;}
  }}\kern3pt}}
}

\newcommand{\klstar}{\text{\kreuz}}  				%

\newcommand{\into}{\rightsquigarrow}

\newcommand\sepimp{\mathrel{-\mkern-6mu\star}}

\newcommand{\geoF}{\mathfrak{f}}
\newcommand{\geoU}{\mathfrak{u}}

\makeatletter
\def\d@r#1#2{\m@th\hfil\rotatebox[origin=c]{-90}{$#1#2$}\mkern3mu}
\renewcommand{\dar}{\mathbin{\mathpalette\d@r\to}}
\makeatother

\newcommand{\ucl}{\oname{cl}}   %
\newcommand{\pro}{\mathrm{prop}}  

\newcommand{\Poset}{\catname{Poset}}

\newcommand{\lrule}[3]{\textbf{(#1)}~~\frac{#2}{#3}}

\newcommand{\infrule}[2]{\frac{#1}{#2}}
\newcommand{\anonrule}[3]{\infrule{#2}{#3}}

\renewcommand{\case}[3]{\mfix{case}{\mathbin{}#1}{of}{#2}{\kern-1pt;}{\mathbin{}#3}}
\DeclareRobustCommand{\pcase}[2]{\mfix{case\,}{#1}{\,of\,}{#2}}

\renewcommand{\ul}[1]{\mkern2mu\underline{\mkern-2mu #1\mkern-2mu}\mkern2mu } %

\usepackage{booktabs} %

\spnewtheorem{thm}[theorem]{Theorem}{\bfseries}{\itshape}
\spnewtheorem{cor}[theorem]{Corollary}{\bfseries}{\itshape}
\spnewtheorem{cnj}[theorem]{Conjecture}{\bfseries}{\itshape}
\spnewtheorem{lem}[theorem]{Lemma}{\bfseries}{\itshape}
\spnewtheorem{lemdefn}[theorem]{Lemma and Definition}{\bfseries}{\itshape}
\spnewtheorem{prop}[theorem]{Proposition}{\bfseries}{\itshape}
\spnewtheorem{defn}[theorem]{Definition}{\bfseries}{\upshape}
\spnewtheorem{rem}[theorem]{Remark}{\bfseries}{\upshape}
\spnewtheorem{notation}[theorem]{Notation}{\bfseries}{\upshape}
\spnewtheorem{expl}[theorem]{Example}{\bfseries}{\upshape}
\spnewtheorem{thmdefn}[theorem]{Theorem and Definition}{\bfseries}{\itshape}
\spnewtheorem{propdefn}[theorem]{Proposition and Definition}{\bfseries}{\itshape}
\spnewtheorem{assumption}[theorem]{Assumption}{\bfseries}{\upshape}
\spnewtheorem{algorithm}[theorem]{Algorithm}{\bfseries}{\upshape}

 \renewenvironment{theorem}{\begin{thm}}{\end{thm}}
 
 \renewenvironment{lemma}{\begin{lem}}{\end{lem}}
 \renewenvironment{proposition}{\begin{prop}}{\end{prop}}
 \renewenvironment{definition}{\begin{defn}}{\end{defn}}
 \renewenvironment{remark}{\begin{rem}}{\end{rem}}
 \renewenvironment{example}{\begin{expl}}{\end{expl}}

\allowdisplaybreaks

\makeatletter
\RequirePackage[bookmarks,unicode,colorlinks=true]{hyperref}%
   \def\@citecolor{blue}%
   \def\@urlcolor{blue}%
   \def\@linkcolor{blue}%

\def\orcidID#1{\smash{\href{http://orcid.org/#1}{\protect\raisebox{-1.25pt}{\protect\includegraphics{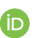}}}}}
\makeatother

\renewcommand{\G}{\Gamma}

\renewcommand{\c}{\colon}
\renewcommand{\l}{\lambda}

\begin{document}

\title{Local Local Reasoning:\\ A BI-Hyperdoctrine for Full Ground Store\thanks{Sergey Goncharov acknowledges support by German Research Foundation (DFG) under project GO~2161/1-2.}}  

\author{Miriam Polzer\and Sergey Goncharov\textsuperscript{(\Letter)}\orcidID{0000-0001-6924-8766} }

\institute{FAU Erlangen-N\"urnberg, Erlangen, Germany\\
\email{\{miriam.polzer,sergey.goncharov\}@fau.de}}

\authorrunning{M.~Polzer, S.~Goncharov}

\maketitle

\begin{abstract}
Modelling and reasoning about dynamic memory allocation is one of the well-established
strands of theoretical computer science, which is particularly well-known as a source
of notorious challenges in semantics, reasoning, and proof theory. We capitalize 
on recent progress on categorical semantics of \emph{full ground store}, 
in terms of a \emph{full ground store monad}, to build a corresponding  
semantics of a higher order logic over the corresponding programs. Our main result
is a construction of an \emph{(intuitionistic) BI-hyperdoctrine}, which is arguably 
the semantic core of higher order logic over local store. Although we have made an 
extensive use of the existing generic tools, certain principled changes
had to be made to enable the desired construction: while the original monad works
over total heaps (to disable dangling pointers), our version involves partial 
heaps (\emph{heaplets}) to enable compositional reasoning using separating 
conjunction. Another remarkable feature of our construction is that, in contrast 
to the existing generic approaches, our BI-algebra does not directly stem from
an internal categorical partial commutative~monoid. 
\end{abstract}

\section{Introduction}\label{sec:intro}
Modelling and reasoning about dynamic memory allocation is a
sophisticated subject in denotational semantics with a long history
(e.g.~\cite{Reynolds97,Oles82,OHearnTennent92,PlotkinPower02}).
Denotational models for dynamic references vary over a large spectrum, and in
fact, in two dimensions: depending on the expressivity of the features being
modelled (\emph{ground store} -- \emph{full ground store} -- \emph{higher order
store}) and depending on the amount of \emph{intensional} information included
in the model (\emph{intensional} -- \emph{extensional}), using the terminology
of Abramsky~\cite{Abramsky14}.

Recently, Kammar et
al~\cite{KammarLevyEtAl17} constructed an extensional monad-based denotational model of the 
\emph{full ground store}, i.e.\ permitting not only memory allocation 
for discrete values, but also storing mutually linked data. The key idea of the latter work is an
explicit delineation between the target presheaf category $[\BW,\Set]$ on which the full ground store
monad acts, and an auxiliary presheaf category $[\BE,\Set]$ of \emph{initializations}, naturally hosting a  
\emph{heap functor} $H$. The latter category also hosts a \emph{hiding monad} $P$,
which can be loosely understood as a semantic mechanism for idealized garbage 
collection. The full ground store monad is then assembled according to the scheme given in Fig.~\ref{fig:main-dia}. As a slogan:
the \emph{local} store monad is a \emph{global} store monad transform of the hiding 
monad sandwiched within a geometric morphism.

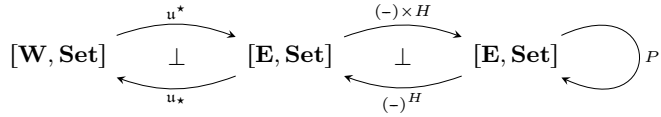
\begin{figure}[t]
\begin{center}
	\begin{tikzcd}[column sep=huge]
		{[\BW, \Set]}
				\arrow[r, bend left=20, "\geoU^\star"]
				\arrow[r, phantom, "\bot"] &
			{[\BE, \Set]}
				\arrow[l, bend left=20, "\geoU_\star "]
				\arrow[r, bend left=20, "(\argument)\times H"]
				\arrow[r, phantom, "\bot"]&
			{[\BE, \Set]}
				\arrow["\Hid", out=25, in=-25, distance=1.5cm, looseness=0.3, loop]
				\arrow[l, bend left=20, "(\argument)^{H}"]
	\end{tikzcd}
\end{center}
	\caption{Construction of the full ground store monad.}
	\label{fig:main-dia}
\end{figure}

The fundamental reason, why extensional models of local store involve intricate 
constructions, such as presheaf categories is that the desirable program equalities
include 
\begin{flalign*}
\qquad&\mbind{\ell\ass\nnew v;\, \ell'\ass\nnew w}{p} ~=~ \mbind{\ell'\ass\nnew w;\, \ell\ass\nnew v}{p} && (\ell\nequiv \ell')\\
\qquad&\mbind{\ell\ass\nnew v}{\ret\star}          ~=~ \ret\star&&\\
\qquad&\mbind{\ell\ass\nnew v}{~(\ite{\ell = \ell'}{\ttrue}{\ffalse})} ~=~ \ffalse &&(\ell\nequiv\ell')
\end{flalign*}
and these jointly do not have set-based models
over countably infinite sets of locations~\cite[Proposition~6]{Staton13}. 
The first equation expresses irrelevance of the memory allocation order,
the second expresses the fact that an unused cell is always garbage collected
and the third guarantees that allocation of a fresh cell does indeed produce a cell 
different from any other. The aforementioned construction validates these
equations and enjoys further pleasant properties, e.g.\ soundness and adequacy 
of a higher order language with user defined storable data structures. 

The goal of our present work is to complement the semantics of programs over 
local store with a corresponding principled semantics of \emph{higher order logic}. 
In order to be able to specify and reason modularly about local store,
more specifically, we seek a model of higher order \emph{separation logic}~\cite{Reynolds02}.
It has been convincingly argued in previous work on categorical models of separation
logic~\cite{BieringBirkedalEtAl07,BizjakBirkedal18} that a core abstraction device
unifying such models is a notion of \emph{BI-hyperdoctrine}, extending
Lawvere's hyperdoctrines~\cite{Lawvere69}, which provide a corresponding abstraction for 
the first order logic. BI-hyperdoctrines are standardly built on \emph{BI-algebras},
which are also standardly constructed from \emph{partial commutative monoids (pcm)}, or more generally
from \emph{resource algebras} as in the \textsc{Iris} state of the art 
advanced framework for higher order separation logic~\cite{JungKrebbersEtAl18}. 
One subtlety our construction reveals is that it does not seem to be possible 
to obtain a \emph{BI-algebra} following general recipes from a pcm (or a resource
algebra), due to the inherent local nature of the storage model, which does not allow
one to canonically map store contents into a global address space.    
Another subtlety is that the devised logic is 
necessarily non-classical, which is intuitively explained by the fact that the 
semantics of programs must be suitably irrelevant to garbage collection, and in our 
case this follows from entirely formal considerations (Yoneda lemma). It is also 
worth mentioning that for this reason the logical theory that we obtain is incompatible 
with the standard (classical or intuitionistic) predicate logic. E.g.\ the formula
$\exists\ell.\,\ell\ito 5$ is always valid in our setup, which expresses the fact 
that a heap \emph{potentially} contains a cell equal to $5$ (which need not be reachable) -- this 
is in accord with the second equation above -- and correspondingly, the formula 
$\forall\ell.\,\neg (\ell\ito 5)$ is unsatisfiable. This and other similar phenomena
are explained by the fact that our semantics essentially behaves as a Kripke semantics 
along two orthogonal axes: (proof relevant)\emph{cell allocation} and 
(proof irrelevant)\emph{cell accessibility}. While the latter captures a \emph{programming}
view of locality, the latter captures a \emph{reasoning} view of locality, and as we 
argue (e.g.\ Example~\ref{expl:impl}), they are generally mutually irreducible.

\paragraph{Related previous work} As we already pointed out, we take inspiration from the 
recent categorical approaches to modelling program semantics for dynamic 
references~\cite{KammarLevyEtAl17}, as well as from higher order separation logic 
semantic frameworks~\cite{BieringBirkedalEtAl07}. Conceptually, the problem of combining 
separation logic with garbage collection mechanisms goes back to Reynolds~\cite{Reynolds00},
who indicated that standard semantics of separation logic is not compatible 
with garbage collection, as we also reinforce with our construction. Calcagno et 
al~\cite{CalcagnoOHearnEtAl03} addressed this issue by providing two models. The 
first model is based on total heaps, featuring the aforementioned effect of ``potential'' 
allocations. To cope with heap separation the authors introduced another model based on partial 
heaps, in which this effect again disappears, and has to be compensated by syntactic 
restrictions on the assertion language.

\paragraph{Plan of the paper} After preliminaries (Section~\ref{sec:prelim}), 
we give a modified presentation of a call-by-value language with full ground references 
and the full ground store monad (Sections~\ref{sec:lang} and~\ref{sec:fgs}) following the 
lines of~\cite{KammarLevyEtAl17}. In Section~\ref{sec:bi} we provide some general 
results for constructing semantics of higher order separation logics. The main 
development starts in Section~\ref{sec:sep} where we provide a construction of a 
BI-hyperdoctrine. We show some example illustrating our semantics in 
Section~\ref{sec:exam} and draw conclusions in Section~\ref{sec:conc}.

This is an extended version of our conference paper~\cite{PolzerGoncharov20}. All 
omitted proofs are collected in appendix.

\section{Preliminaries}\label{sec:prelim}
We assume basic familiarity with the elementary concepts of category
theory~\cite{MacLane71,Johnstone02}, all the way up to monads,  toposes, (co)ends and Kan extensions.
We denote by $|\BC|$ the class of objects of a category~$\BC$;
we often suppress subscripts of natural transformation components if no confusion
arises.

In this paper, we work with special kinds of \emph{covariant presheaf toposes},
i.e.\ functor categories of the form $[\BC,\Set]$, where $\BC$ is small and satisfies
the following \emph{amalgamation condition}: for any $f\c a\to b$ and $g\c a\to c$ there
exist $g'\c b\to d$ and $f'\c c\to d$ such that $f'\comp g = g'\comp f$. Such toposes
are particularly well-behaved, and fall into the more general class of \emph{De~Morgan}
toposes~\cite{Johnstone79}. As presheaf toposes, De~Morgan toposes are precisely
characterized by the condition that ${2=1+1}$ is a retract of the subobject
classifier $\Omega$. More specifically, our~$\BC$ support further useful
structure, in particular, a strict monoidal tensor $\oplus$ with jointly
epic injections $\inj_1$, $\inj_2,$ forming an \emph{independent coproduct} structure,
as recently identified by Simpson~\cite{Simpson18}.
Moreover, if the coslices $c\,\dar\BC$ support independent products, we obtain
\emph{local independent coproducts} in $\BC$, which are essentially cospans
$c_1\to c_1\oplus_c c_2\from c_2$ in $c\,\dar\BC$. Given
$\rho_1\c c\to c_1$ and $\rho_2\c c\to c_2$, we thus always have $\rho_1\bullet\rho_2\c c_1\to c_1\oplus_c c_2$
and $\rho_2\bullet\rho_1\c S_2\to c_1\oplus_c c_2$, such that $(\rho_1\bullet\rho_2)\comp\rho_1=(\rho_2\bullet\rho_1)\comp\rho_2$,
and as a consequence, $[\BC,\Set]$ is a De~Morgan topos. Intuitively, the category
$\BC$ represents worlds in the sense of \emph{possible world semantics}~\cite{Oles82,Reynolds97}.
A morphism $\rho\c a\to b$ witnesses the fact that $b$ is a \emph{future} world w.r.t.\
$a$. Existence of local independent products intuitively ensures that diverse
futures of a given world can eventually be unified in a canonical way.

Every functor $\geoF\c\BC\to\BD$ induces a functor $\geoF^\star\c[\BD,\Set]\to [\BC,\Set]$
by precomposition with $\geoF$. By general considerations, there is a right
adjoint $\geoF_\star\c[\BC,\Set]\to [\BD,\Set]$, computed as $\oname{Ran}_\geoF$,
the right Kan extension along~$\geoF$. This renders the adjunction $\geoF^\star\dashv \geoF_\star$,
as a \emph{geometric morphism}, in particular, $\geoF^\star$ preserves all
finite limits.

\begin{figure}[t!]
\begin{center}%
{\parbox{\textwidth}{%
\small%
\begin{subfigure}{\textwidth}
\begin{flalign*}
 \lrule{var}{%
		  x\c A\text{~~~in~~~}\G
	  }{%
		 \G\vctx x\c A
	  }
&&
\lrule{inl}{%
		\G\vctx v\c A 
	}{%
		\G\vctx\inl v\c A+B
	}
&&
\lrule{inr}{%
		\G\vctx w\c B 
	}{%
		\G\vctx\inr w\c A+B
	}
\end{flalign*}\\[-5ex]
\begin{flalign*}
\lrule{unit}{}{\G\vctx\star\c 1}
&&
\lrule{pair}{%
		\G\vctx v\c A\quad\G\vctx w\c B 
	}{%
		\G\vctx (v,w)\c A\times B
	}
&&
\lrule{abs}{%
		\G, x\c A\cctx p\c B
	}{%
		\G\vctx\l x.\, p\c A\to B
	}
\end{flalign*}
\medskip
\end{subfigure}

\dotfill

\medskip
\begin{subfigure}{\textwidth}
\begin{flalign*}
 \lrule{c-case}{%
 \G\vctx v\c A+B\qquad
 \G,x\c A\cctx p\c C\qquad\G,y\c B\cctx q\c C
  }{%
 \G\cctx\case{v}{\inl x\mto p}{
   \inr y\mto q}\c C
  }
\end{flalign*}\\[-4ex]
\begin{flalign*}
 \lrule{p-case}{%
 \G\vctx v\c A\times B\quad
 \G,x\c A,y\c B\cctx q\c C
  }{%
 \G\cctx\pcase{v}{(x,y)\mto q}\c C
  }
&&
\lrule{init}{\G\vctx v\c 0}{\G\cctx\oname{init} v\c A}
\end{flalign*}\\[-4ex]
\begin{flalign*}
&&
\lrule{ret}{%
		\G\vctx v\c A 
	}{%
		\G\cctx\ret v\c A
	}
&&
\lrule{let}{%
\G\cctx p\c A\qquad \G,x\c A\cctx q\c B   
}
{
\G\cctx\mbind{x\ass p}{q}\c B 
}
&&
\end{flalign*}\\[-4ex]
\begin{flalign*}
 \lrule{app}{%
 \G\vctx v\c A\to B\qquad\G\vctx w\c A
  }{%
 \G\cctx v w\c B
  }
\end{flalign*}\\[-4ex]
\begin{flalign*}
\lrule{put}{%
  \G\vctx \ell\c\Ref_S\qquad\G\vctx v\c\CType(S)
  }{%
  \G\cctx \ell\ass v\c 1 
  }
&&
\lrule{get}{%
  \G\vctx \ell\c\Ref_S
  }{%
  \G\cctx\bang \ell\c \CType(S)
  }
\end{flalign*}\\[-5ex]
\begin{flalign*}
&&
\lrule{new}{%
\begin{array}{r@{}l}
  \G, \ell_1\c\Ref_{S_1},\ldots, \ell_n\c\Ref_{S_n} &\vctx v_1\c\CType(S_1)\\
    &\vdots\\
  \G, \ell_1\c\Ref_{S_1},\ldots, \ell_n\c\Ref_{S_n} &\vctx v_n\c\CType(S_n)\\[1ex]
  \G, \ell_1\c\Ref_{S_1},\ldots, \ell_n\c\Ref_{S_n} &\cctx p\c A
\end{array}
  }{%
  \G\cctx\letref{\ell_1\ass v_1,\ldots,\ell_n\ass v_n}{p}\c A 
  }
&&
\end{flalign*}
\end{subfigure}
}}
\end{center}
\caption{Term formation rules for values (top) and computations (bottom).}
\label{fig:prog.lang}
\end{figure}

\section{A Call-by-Value Language with Local References}\label{sec:lang}
To set the context, we consider the following higher order language of programs with local references
by slightly adapting the language of Kammar et al~\cite{KammarLevyEtAl17} to match 
with the \emph{fine-grain call-by-value} perspective~\cite{LevyPowerEtAl02}. This allows us to formally 
distinguish \emph{pure} and \emph{effectful} judgements.
First, we postulate a collection of \emph{cell sorts} $\CS$ and then introduce 
further types with the grammar:
\begin{flalign}\label{eq:gramma}
&&A,B\ldots\Coloneqq  0\mid 1\mid A\times B\mid A+B\mid A\to B\mid\Ref_S && (S\in\CS) &&
\end{flalign}
A type is \emph{first order} if it does not involve the function type
constructors $A\to B$. We then fix a map~$\CType$, assigning a first order type
to every given sort from~$\CS$. The corresponding term formation rules over
these data are given in Fig.~\ref{fig:prog.lang} where the $\mathsf{v}$-indices 
at the turnstiles indicate \emph{values} and the $\mathsf{c}$-indices indicate 
\emph{computations}. The only non-standard rules~\textbf{(put)},~\textbf{(get)}
and~\textbf{(new)} are expected to handle references in the expected way: \textbf{(put)}
updates the cell referenced by $\ell$ with a value $v$,~\textbf{(get)} returns a value
under the reference $\ell$ and~\textbf{(new)} simultaneously allocates new cells filled
with the values $v_1,\ldots,v_n$ and makes them accessible in $p$ under the corresponding 
references $\ell_1,\ldots,\ell_n$. As a fine-grain call-by-value language, the language 
in Fig.~\ref{fig:prog.lang} can be interpreted in a standard way over a category with
a strong monad, as long as the latter can provide a semantics to the rules~\textbf{(put)},~\textbf{(get)} 
and~\textbf{(new)}. We present this monad in detail in the next section.
\begin{example}[Doubly Linked Lists]
Let $\CS = \{\mathit{DLList}\}$ and let
$\CType(\mathit{DLList}) = 2\times(\Ref_{\mathit{DLList}}+1)
\times(\Ref_{\mathit{DLList}}+1)$, which indicates that a list
element is a Boolean (i.e.\ an element of $2=1+1$) and two pointers (forwards and 
backwards) to list elements, each of which may be missing. Note that we thus avoid
empty lists and null-pointers: every list contains at least one element, and the 
elements added by $+1$ cannot be dereferenced. This example provides 
a suitable illustration for the $\oname{letref}$ construct. E.g.\ the program
\begin{align*}
\letref{\ell_1\ass (0,\inr\star,\inl \ell_2);\, \ell_2\ass (1,\inl \ell_1,\inr\star)}{\ret \ell_1}
\end{align*}
simultaneously creates two list elements pointing to each other and returns 
a reference to the first one.
\end{example}

\section{Full Ground Store in the Abstract}\label{sec:fgs}
We proceed to present the full ground store 
monad by slightly tweaking the original construction~\cite{KammarLevyEtAl17} 
towards higher generality. The main distinction is that we do not recur to any specific program syntax
and proceed in a completely axiomatic manner in terms of functors
and natural transformations. This mainly serves the purpose of developing our 
logic in Section~\ref{sec:sep}, which will require a coherent upgrade of the 
present model. Besides this, in this section we demonstrate flexibility of 
our formulation by showing that it also instantiates to the model previously 
developed by Plotkin and Power~\cite{PlotkinPower02} (Theorem~\ref{thm:simp-store}). 

Our present formalization is parametric in three aspects: the
set of \emph{sorts}~$\CS$, the set of \emph{locations}~$\CL$ and a map $\oname{range}$,
introduced below for interpreting~$\CS$. We assume that~$\CL$ is canonically isomorphic
to the set of natural numbers $\nat$ under $\hash\c\CL\cong\BBN$. Using this
isomorphism, we commonly use the ``shift of $\ell\in\CL$ by $n\in\nat$'', defined as follows:
$\ell + n = \hash^\mone(\hash\ell + n)$.

\paragraph{Heap layouts and abstract heap(let)s}
Let $\BW$ be a category of \emph{(heap) layouts} and injections defined as follows:
an object $w \in |\BW|$ is a finitely supported partial function
$w\c\CL\rightharpoonup_{\mathit{fin}} \CS$ and a morphism $\rho\c w\to w'$ is a type preserving injection $\rho\c\dom
w \to\dom w'$, i.e. for all $l \in\img w$, $w(\ell) = w'(\rho(\ell))$.
We will equivalently view $w$ as a left-unique subset of $\CL \times \CS$ and hence
use the notation $(\ell\c S)\in w$ as an equivalent of $w(\ell) = S$. 
Injections $\rho\c w\to w'$
with the property that $w(\ell\c S) = \ell\c S$ for all $(\ell\c S)\in w$ we
also call \emph{inclusions} and write $w\subseteq w'$ instead of
$\rho\c w\to w'$, for obviously there is at most one inclusion from $w$
to $w'$. 
If $w\subseteq w'$ then we call $w$ a \emph{sublayout} of $w'$.
We next postulate
\begin{align*}
	\oname{range}\c\CS \to [\BW, \Set].
\end{align*}
The idea is, given a sort $S \in\CS$ and a heap layout $w \in |\BW|$,
$\oname{range}(S)(w)$ yields the set of possible values for cells of type $S$ over $w$.
\begin{example}\label{ex:main}
Assuming the grammar~\eqref{eq:gramma} and a corresponding 
map $\CType$, a generic type~$A$ is interpreted as a presheaf $\ul{A}\c\BW\to\Set$, 
by obvious structural induction, e.g.\ $\ul{A\times B} = \ul{A}\times\ul{B}$, except 
for the clause for $\Ref$, for which $(\ul{\Ref_S}) w = w^{\mone}(S)$. 
This yields the following definition for $\oname{range}$: $\oname{range}(S)= \ul{\CType(S)}$~\cite{KammarLevyEtAl17}.
\end{example}
\begin{example}[Simple Store]\label{def:gs}
By taking $\CS = \{ \star \}$, $\CL=\nat$ (natural numbers) and $\oname{range} (\star) (w) = \CV$ 
where $\CV$ is a fixed set of \emph{values}, we essentially obtain the model 
previously explored by Plotkin and Power~\cite{PlotkinPower02}. 
We reserve the term \emph{simple store} for this instance.
Simple store is 
a ground store (since $\oname{range}$ is a constant functor), moreover this store 
is untyped (since $\CS = \{\star\}$) and the locations $\CL$ are precisely the 
natural numbers.
\end{example}
A \emph{heap} over a layout $w$ assigns to each $(\ell\c S)\in w$ an element from
$\oname{range}(S)(w)$. More generally, a $\emph{heaplet}$ over  $w$ assigns
an element from $\oname{range}(S)(w)$ to \emph{some}, possibly not all,
$(\ell\c S)\in w$. We thus define the following \emph{heaplet bi-functor} $\Hplb\c\BW^{\op} \times \BW \to\Set$:
\begin{align*}
	\Hplb(w^-, w^+) = \prod_{(\ell\c S)\in w^-} \oname{range}(S)(w^+)
\end{align*}
and identify the elements of $\Hplb(w^-, w^+)$ with heaplets and the elements of
$\Hplb(w, w)$ with heaps. Of course, we intend to use $\Hplb(w^-, w^+)$ for such
$w^-$ and~$w^+$ that the former is a sublayout of the latter.
The contravariant action of $H$ is given by projection and the covariant
action is induced by functoriality of $\oname{range}(S)$.
\begin{align*}
\pr_{(\ell\c S)}(\Hplb(w^-, \rho_1\c w^+_1\to w^+_2)(\eta\in\Hplb(w^-,w_1^+)))  &\,= \oname{range}(S)(\rho_1) (\pr_{(\ell\c S)} \eta)\\*
\pr_{(\ell\c S)}(\Hplb(\rho_2\c w^-_2\to w^-_1, w^+)(\eta\in\Hplb(w_1^-, w^+)))    &\,= \pr_{\rho_2(\ell\c S)} \eta 
\end{align*}
The heaplet functor preserves independent coproduct, we overload the $\oplus$
operation with the isomorphism $\oplus\c \Hplb(w_1, w)\times \Hplb(w_2, w)\iso\Hplb(w_1 \oplus w_2, w)$.

\begin{example}\label{ex:ref-N}
For illustration, consider the following simplistic example. Let $\CS = \{
\mathit{Int}\comma\Ref_{\mathit{Int}}\comma\Ref_{\Ref_{\mathit{Int}}}\comma \dots \}$
where $\mathit{Int}$ is meant to capture the ground type of integers and recursively,
$\Ref_{A}$ is the type of pointers to $A$. Then, we put
\begin{align*}
	&\oname{range} (\mathit{Int})(w) = \BBZ,&
	&\oname{range} (\Ref_{S})(w) = w^{\mone}(S) = \{\ell\in\dom w \mid
	 w(\ell) = S \}.
\end{align*}
For a heaplet example, consider
		$w^- = \{\ell_1\c\mathit{Int}, \ell_2\c\Ref_{\mathit{Int}} \}$ and
		$w^+ = \{\ell_1\c\mathit{Int}, \ell_2\c\Ref_{\mathit{Int}}, \ell_3\c\mathit{Int} \}$.
Hence,~$w^-$ is a sublayout of $w^+$. By viewing the elements of $\Hplb(w^-, w^+)$ as
lists of assignments on $w^-$, we can define $s_1, s_2 \in\Hplb(w^-, w^+)$ as follows:
		$s_1 = [\ell_1\c\mathit{Int} \mto 5\comma \ell_2\c\oname{Ref_{int}} \mto \ell_1]$, $s_2 = [\ell_1\c\mathit{Int} \mto 3\comma \ell_2\c\oname{Ref_{int}} \mto \ell_3]$.
The heaplets	$s_1$ and $s_2$ can be graphically presented as follows:
	\begin{center}
		\begin{tikzpicture}[scale=.8]
			\draw (0, 0) rectangle (1, 0.5) node[pos=0.5] {$5$};
			\draw (0, 0.5) rectangle (1, 1);
				\draw[->, thick, sigma line] (0.5,0.75) node {$\bullet$} to (3.5,0.25);
			\node at (0.5, -0.5) {$w^-$};

			\draw (3, 0) rectangle (4, 0.5);
			\draw (3, 0.5) rectangle (4, 1);
			\draw (3, 1) rectangle (4, 1.5);
			\node at (3.5, -0.5) {$w^+$};
		\end{tikzpicture}
		\hspace{4em}
		\begin{tikzpicture}[scale=.8]
			\draw (0, 0) rectangle (1, 0.5) node[pos=0.5] {$3$};
			\draw (0, 0.5) rectangle (1, 1);
				\draw[->, thick, sigma line] (0.5,0.75) node {$\bullet$} to (3.5,1.25);
			\node at (0.5, -0.5) {$w^-$};

			\draw (3, 0) rectangle (4, 0.5);
			\draw (3, 0.5) rectangle (4, 1);
			\draw (3, 1) rectangle (4, 1.5);
			\node at (3.5, -0.5) {$w^+$};
		\end{tikzpicture}
	\end{center}
\end{example}
The category $\BW$ supports (local) independent coproducts described in
Section~\ref{sec:prelim}. These are constructed as follows. For $w,w'\in |\BC|$,
$w\oplus w' = w\cup \{\ell+n+1\c S\mid (\ell,c)\in w'\}$ with $n$ being the largest
index for which $w$ is defined on $\hash^\mone(n)$. This yields a strict
monoidal structure $\oplus\c\BW \times \BW \to\BW$. Intuitively, $w_1 \oplus w_2$
is a canonical disjoint sum of $w_1$ and $w_2$, but note that $\oplus$ is not
a coproduct in~$\BW$ (e.g.\ there is no $\nabla\c 1 \oplus 1 \to 1$, for $\BW$
only contains injections).
For every $\rho\c w _1 \to w_2$, there is a canonical complement $\rho^\complement\c w_2 \ominus \rho\to w_2$ whose domain
${{w_2} \ominus \rho} = {w_2} \smin\img\rho$ consists of all such cells $(\ell\c S)\in w_2$ that~$\rho$ misses.
Given two morphisms $\rho_1 \c w\to w_1$ and $\rho_2 \c w\to w_2$, we define the
local independent coproduct $w_1\oplus_w w_2$ as the layout consisting of the
locations from~$w$, and the ones from $w_1$ and $w_2$ which are neither in the
image of $\rho_1$ nor in the image of~$\rho_2$:
\begin{align*}
	\rho_1 \oplus_w \rho_2 = w \oplus (w_1 \ominus \rho_1)\oplus (w_2 \ominus \rho_2).
\end{align*}
There are morphisms $w_1 \xto{\rho_1\bullet\rho_2} \rho_1 \oplus_w \rho_2$ and
$w_2 \xto{\rho_2\bullet\rho_1} \rho_1 \oplus_w \rho_2$ such that

\begin{center}
	\begin{tikzcd}[column sep=huge]
		w \rar{\rho_2} \dar[swap]{\rho_1} & w_2 \dar{\rho_2\bullet\rho_1}\\
		w_1 \rar{\rho_1\bullet\rho_2} & \rho_1 \oplus_w \rho_2
	\end{tikzcd}
\end{center}
Fig.~\ref{fig:ind-coproducts} illustrates this definition with a concrete
example.

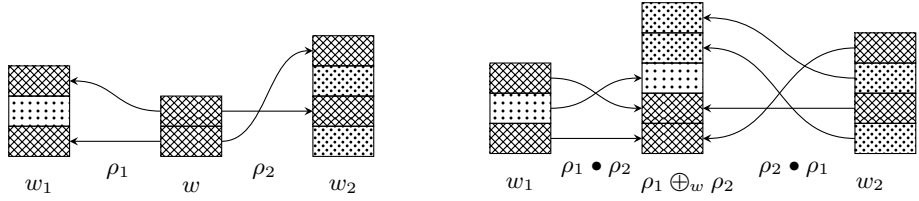
\begin{figure}[t]
	\begin{center}
	\begin{tikzpicture}[scale=.8]
		\draw [pattern=crosshatch] (0, 0) rectangle (1, 0.5);
		\draw [pattern=dots] (0, 0.5) rectangle (1, 1);
		\draw [pattern=crosshatch] (0, 1) rectangle (1, 1.5);
		\node at (0.5, -0.5) {$w_1$};

		\draw[->, sigma line] (2.5, 0.25) to (1,0.25);
		\draw[->, sigma line] (2.5, 0.75) to (1,1.25);
		\node at (1.8, -0.25) {$\rho_1$};

		\draw [pattern=crosshatch] (2.5, 0) rectangle (3.5, 0.5);
		\draw [pattern=crosshatch] (2.5, 0.5) rectangle (3.5, 1);
		\node at (3, -0.5) {$w$};

		\draw[->, sigma line] (3.5, 0.25) to (5, 1.75);
		\draw[->, sigma line] (3.5, 0.75) to (5, 0.75);
		\node at (4.2, -0.25) {$\rho_2$};

		\draw [pattern=crosshatch dots] (5, 0) rectangle (6, 0.5);
		\draw [pattern=crosshatch] (5, 0.5) rectangle (6, 1);
		\draw [pattern=crosshatch dots] (5, 1) rectangle (6, 1.5);
		\draw [pattern=crosshatch] (5, 1.5) rectangle (6, 2);
		\node at (5.5, -0.5) {$w_2$};
	\end{tikzpicture}
  \hspace{4em}
	\begin{tikzpicture}[scale=.8]
		\draw [pattern=crosshatch] (0, 0)   rectangle (1, .5);
		\draw [pattern=dots]      (0, 0.5) rectangle (1, 1);
		\draw [pattern=crosshatch] (0, 1)   rectangle (1, 1.5);
		\node at (0.5, -0.5) {$w_1$};

		\draw[->, sigma line] (1, 0.25) to (2.5, 0.25);
		\draw[->, sigma line] (1, 0.75) to (2.5, 1.25);
		\draw[->, sigma line] (1, 1.25) to (2.5, 0.75);
		\node at (1.75, -0.25) {$\rho_1\bullet\rho_2$};

		\draw [pattern=crosshatch] (2.5, 0)         rectangle (3.5, .5);
		\draw [pattern=crosshatch] (2.5, 0.5)       rectangle (3.5, 1);
		\draw [pattern=dots] (2.5, 1)               rectangle (3.5, 1.5);
		\draw [pattern=crosshatch dots] (2.5, 1.5)  rectangle (3.5, 2);
		\draw [pattern=crosshatch dots] (2.5, 2)    rectangle (3.5, 2.5);
		\node at (3.25, -0.5) {$\rho_1 \oplus_w \rho_2$};

		\draw[->, sigma line] (6, 0.25) to (3.5, 1.75);
		\draw[->, sigma line] (6, 0.75) to (3.5, 0.75);
		\draw[->, sigma line] (6, 1.25) to (3.5, 2.25);
		\draw[->, sigma line] (6, 1.75) to (3.5, 0.25);
		\node at (5, -0.25) {$\rho_2\bullet\rho_1$};

		\draw [pattern=crosshatch dots] (6, 0) rectangle (7, 0.5);
		\draw [pattern=crosshatch] (6, 0.5) rectangle (7, 1);
		\draw [pattern=crosshatch dots] (6, 1) rectangle (7, 1.5);
		\draw [pattern=crosshatch] (6, 1.5) rectangle (7, 2);
		\node at (6.25, -0.5) {$w_2$};
	\end{tikzpicture}
	\caption{Local independent coproduct}
	\label{fig:ind-coproducts}
	\end{center}
\end{figure}

\paragraph{Initialization and hiding}
Note that in the simple store model (Definition~\ref{def:gs}),~$\CH$ is equivalently
a contravariant functor $H\c\BW^{\op} \to\Set$ with $H w = \CV^w$, hence~$\CH$ can be placed
e.g.\ in $[\BW^\op,\Set]$. In general,~$\CH$ is mix-variant, which calls for
a more ingenious category where $H$ could be placed. Designing such category
is indeed the key insight of~\cite{KammarLevyEtAl17}. Closely following this work, 
we introduce a category~$\BE$, whose objects are the same as
those of $\BW$, and the morphisms $\epsilon\in\BE(w, w')$,
called \emph{initializations}, consist of an injection $\rho\c w\to w'$ and a heaplet $\eta \in\Hplb(w'\ominus \rho, w')$:
\begin{align*}
	\BE(w, w') = \sum_{\rho\c w\to w'} \Hplb(w'\ominus \rho, w').
\end{align*}
Recall that the morphism $\rho\c w\to w'$ represents a move from a world 
with~$w$ allocated memory cells a world with $w'$ allocated memory cells.
A morphism of $\BE$ is a morphism of $\BW$ augmented with a heaplet part~$\eta$, 
which provides the information how the newly allocated
cells in $w'\ominus\rho$ are filled. The heap functor now can be viewed as a representable
presheaf $\Hpl\c\BE \to\Set$ essentially because by definition, $\Hpl w = \Hplb(w,w)\cong \BE(\iobj, w)$.
Let us agree to use the notation $\eps\c w\into w'$ for morphisms in $\BE$
to avoid confusion with the morphisms in $\BW$.

Like $\BW$, $\BE$ supports local independent coproducts, but remarkably $\BE$ does
not have vanilla independent coproducts, due to the fact that $\BE$ does not
have an initial object. That is, in turn, because defining an inital morphism
would amount to defining canonical fresh values for newly allocated cells, but
those need not exist. The local independent coproducts of $\BW$ and $\BE$ agree in the sense that we
can \emph{promote} an initialization $(\rho_2, \eta)\c w\into w_2$ along an injection
$\rho_1 \c w \to w_1$ to obtain an initialization
$\rho_1\bullet (\rho_2, \eta)\c w_1\into\rho_1 \oplus_{w_1} \rho_2$.
This is accomplished by mapping the heaplet structure $\eta$
forward along $\rho_2\bullet\rho_1 \c w _2 \to\rho_1 \oplus_{w} \rho_2$.

\paragraph{Hiding monad}
Recall that the local store is supposed to be insensitive to garbage 
collection. This is captured by identifying the stores that agree on their 
observable parts using the \emph{hiding monad} $\Hid$ defined on $[\BE,\Set]$ as follows:
\begin{align}\label{eq:hid}
	(\Hid X) w = \int^{\rho\c w\to w'\in w \dar\geoU} X w'.
\end{align}
Here, $\geoU\c\BE \to\BW$ is the obvious heaplet discarding functor $\geoU(\rho, \eta) = \rho$.
Intuitively, in~\eqref{eq:hid}, we view the locations of $w$ as public and the ones of $w'\ominus \rho$
as private.
The integral sign denotes a \emph{coend}, which in this case is just an ordinary colimit on $\Set$
and is computed as a quotient of $\sum_{\rho\c w\to w'\in w \dar\geoU} X w'$
under the equivalence relation $\sim$ obtained as a symmetric-transitive closure 
of the relation
\begin{flalign*}
&&	(\rho\c w\to w_1, x \in X w_1)\preceq (\geoU\epsilon\comp \rho\c w\to w_2, (X \epsilon)(x)\in X w_2) && (\epsilon\c w_1\into w_2)
\end{flalign*}
Note that $\preceq$ is a preorder. Moreover, it enjoys the following \emph{diamond property}.
\begin{proposition}\label{prop:diamond}
If $(\rho,x)\preceq (\rho_1,x_1)$ and $(\rho,x)\preceq (\rho_2,x_2)$ then 
$(\rho_1,x_1)\preceq (\rho',x')$ and $(\rho_2,x_2)\preceq (\rho',x')$ for a suitable 
$(\rho',x')$. 

Hence $(\rho_1,x_1)\sim (\rho_2,x_2)$ iff $(\rho_1,x_1)\preceq (\rho,x)$, 
$(\rho_2,x_2)\preceq (\rho,x)$ for some~$(\rho,x)$.
\end{proposition}

\begin{example}
To illustrate the equivalence relation $\sim$ behind $P$, we revisit the setting
of Example~\ref{ex:ref-N}. Consider the following situations:
	\begin{equation*}
		\begin{tikzpicture}
			\draw [] (0, 0) rectangle (1, 0.5) node[midway] {5};
			\draw [dotted] (0, 0.5) rectangle (1, 1);
			\draw [->, thick] (0.5, 0.75) node {$\bullet$} |- (-.75,0.75) |- (0, 0.25);
			\draw [dotted] (0, 1) rectangle (1, 1.5) node[midway] {6};
		\end{tikzpicture}\hspace{1em}\sim\hspace{1em}
		\begin{tikzpicture}
			\draw [] (0, 0) rectangle (1, 0.5) node[midway] {5};
		\end{tikzpicture} \hspace{5em}
		\begin{tikzpicture}
			\draw [dotted] (0, 0) rectangle (1, 0.5) node[midway] {5};
			\draw [] (0, 0.5) rectangle (1, 1);
			\draw [->, thick] (0.5, 0.75) node {$\bullet$} |- (-0.75,0.75) |- (0, 0.25);
			\draw [] (0, 1) rectangle (1, 1.5) node[midway] {6};
		\end{tikzpicture}
		\hspace{1em}\not\sim\hspace{1em}
		\begin{tikzpicture}
			\draw [dotted] (0, 0) rectangle (1, 0.5) node[midway] {3};
			\draw [] (0, 0.5) rectangle (1, 1);
			\draw [->, thick] (0.5, 0.75) node {$\bullet$} |- (-0.75,0.75) |- (0, 0.25);
			\draw [] (0, 1) rectangle (1, 1.5) node[midway] {6};
		\end{tikzpicture}
	\end{equation*}
Here, the solid lines indicate public locations and the dotted lines indicate private locations.
The left equivalence holds because the private locations are not reachable from
the public ones by references (depicted as arrows). On the right,
although the public parts are equal, the reachable cells of the private parts reveal the distinction,
preventing the equivalence under~$\sim$. Intuitively, 
hiding identifies those heaps that agree both on their public and reachable private part.
\end{example}
The covariant action of $\Hid X$ (on $\BE$) is defined via promotion of initializations:
\begin{align*}
	(\Hid X)(\epsilon\c w_1\into w_2)& (\rho\c w _1 \to w_1', x \in Xw_1')_\sim\\*
 =&\;	(\geoU\epsilon\bullet\rho\c w_2\to\rho\oplus_{w_1} \geoU\epsilon, X(\rho\bullet\epsilon)(x))_\sim.
\end{align*}
Furthermore, there is a contravariant hiding operation (on $\BW$) given by the canonical
action of the coend: for $\rho\c w\to w'$, we define
$ \oname{hide_\rho}\c\Hid X w'\to\Hid X w $:
\begin{align}\label{eq:P-cov-act}
	\oname{hide_\rho} (\rho'\c w '\to w'', x \in Xw'')_\sim~ =
		(\rho'\comp\rho, x)_\sim
\end{align}
This allows us to regard $P$ both as a functor $[\BE,\Set]\to[\BE,\Set]$ and as a functor 
$[\BE,\Set]\to [\BW^\op,\Set]$.

\paragraph{Full ground store monad}
We now have all the necessary ingredients to obtain the full ground store
monad~$\St$ on $[\BW,\Set]$. This monad is assembled by composing the
functors in Fig.~\ref{fig:main-dia} in the following way. First, observe
that $(\Hid(\argument\times H))^H$ is a standard (global) store monad
transform of $\Hid$ on $[\BE,\Set]$. This monad is sandwiched between
the adjunction $\geoU_\star\vdash \geoU^\star$ induced by $\geoU$ (see
Section~\ref{sec:prelim}). Since any monad itself resolves into an adjunction,
sandwiching in it between an adjunction again yields a monad. In summary,
\begin{align}\label{eq:T}
\St = \Bigl([\BW,\Set] \xto{~\geoU^\star} [\BE,\Set]\xto{\Hid(-\times H)^H} [\BE,\Set] \xto{~\geoU_\star} [\BW,\Set]\Bigr).\\[-4ex]\notag
\end{align}
\needspace{2\baselineskip}
\begin{theorem}\label{thm:T-strong}
The monad $T$, defined by~\eqref{eq:T} is strong.
\end{theorem}
\begin{proof}
The proof is a straightforward generalization of the proof in~\cite{KammarLevyEtAl17}.
\qed\end{proof}
We can recover the monad previously developed by Plotkin and Power~\cite{PlotkinPower02}
by resorting to the simple store (Example~\ref{def:gs}).
\begin{theorem}\label{thm:simp-store}
Under the simple store model $\St$ is isomorphic to the local store monad from~\cite{PlotkinPower02}:
\begin{align*}
		(\St X) w \cong \bigg( \int^{\rho\c w\to w'\in w \dar \BW} Xw'\times \CV^{w'} \biggl)^{\CV^{w}}.
\end{align*}
\end{theorem}
Using~\eqref{eq:T}, one obtains the requisite semantics to the language in Fig.~\ref{fig:prog.lang} 
using the standard clauses of fine-grain call-by-value~\cite{LevyPowerEtAl02}, 
except for the special clauses for~\textbf{(put)},~\textbf{(get)} and~\textbf{(new)}, 
which require special operations of the monad:
\begin{flalign*}
\qquad\oname{get}&\c\geoU^\star\ul{\Ref_S}\times H\to\geoU^\star\ul{\CType(S)}{\times H}\\*
\oname{put}&\c (\geoU^\star\ul{\Ref_S}\times \geoU^\star\ul{\CType(S)})\times H\to 1\times H\\*
\oname{new}&\c\geoU^\star(\ul{\CType(S)}^{\ul{\Ref_{S}}})\times H\to P(\geoU^\star\ul{\Ref_{S}}\times H)
\end{flalign*}

\section{Intermezzo: BI-Hyperdoctrines and BI-Algebras}\label{sec:bi}
To be able to give a categorical notion of higher order logic over local store, 
following Biering et~al~\cite{BieringBirkedalEtAl07}, we aim to construct a \emph{BI-hyperdoctrine}.

Note that algebraic structures, such as monoids and Heyting algebras can be 
straightforwardly internalized in any category with finite products, which gives 
rise to \emph{internal monoids}, \emph{internal Heyting algebras}, etc. The situation 
changes when considering non-algebraic properties. In particular, recall that 
a Heyting algebra~$A$ is \emph{complete} iff it has arbitrary joins, which are 
preserved by binary meets. The corresponding categorical notion is essentially obtained 
from spelling out generic definitions from internal category theory~\cite[B2]{Johnstone02} and is as follows.
\begin{definition}[Internally Complete Heyting Algebras]\label{def:ic-HA}
An internal Heyting (Boolean) algebra $A$ in a finitely complete category $\BC$ is 
\emph{internally complete} if for every $f\in\BC(I,J)$, there exist \emph{indexed joins} 
$\bigor_f\c\BC(I,A)\to\BC(J,A)$, left order-adjoint to $(\argument)\comp f\c\BC(J,A)\to\BC(I,A)$
such that 
for any pullback square on the left, the corresponding diagram on the 
right commutes (\emph{Beck-Chevalley condition}): 
\begin{equation*}
\begin{tikzcd}[column sep = 8ex,row sep = 4ex]
I  
\pbk
\dar["g"']
\rar["f"] 
&  
J
\dar["h"]
\\
I'
\rar["f'"'] 
& 
J'
\end{tikzcd}
\hspace{2cm}
\begin{tikzcd}[column sep = 12ex,row sep = 4ex]
\BC(J,A)  
\dar["\bigor_h"']
\rar["(\argument)\comp f"] 
&  
\BC(I,A)
\dar["\bigor_g"]
\\
\BC(J',A)
\rar["(\argument)\comp f'"'] 
& 
\BC(I',A)
\end{tikzcd}
\end{equation*}
\end{definition}
It follows generally that existence of indexed joins $\bigor$ implies existence of 
indexed meets $\bigand$, which then satisfy dual conditions (\cite[Corollary~2.4.8]{Johnstone02}).
\begin{remark}[Binary Joins/Meets]\label{rem:bin-joins}
The adjointness condition for indexed joins means precisely that $\bigor_f\phi\leq\psi$
iff $\phi\leq\psi\comp f$ for every $\phi\c I\to A$ and every $\psi\c J\to A$.
If $\BC$ has binary coproducts, by taking $f=\nabla\c X+X\to X$ we obtain 
that $\bigor_\nabla\phi\leq\psi$ iff $\phi\leq[\psi,\psi]$ iff $\phi\comp\inl\leq\psi$ 
and $\phi\comp\inr\leq\psi$. This characterizes $\bigor_\nabla[\phi_1,\phi_2]\c X\to A$ as the 
binary join of $\phi_1,\phi_2\c X\to A$. Binary meets are characterized analogously.
\end{remark}
\begin{definition}[(First Order) (BI-)Hyperdoctrine]
Let $\BC$ be a category with finite products. A \emph{first order hyperdoctrine
over $\BC$} is a functor $S\c\BC^\op\to\Poset$ with the following properties:
\begin{enumerate}
\item given $X\in |\BC|$, $SX$ is a Heyting algebra;
\item given $f\in\BC(X,Y)$, $Sf\c SY\to SX$ is a Heyting algebra morphism;
\item for any product projection $\fst\c X\times Y\to X$, there are 
$(\exists Y)_X\c S(X\times Y)\to SX$ and $(\forall Y)_X\c S(X\times Y)\to SX$,
which are respective left and right order-adjoints of $S\fst\c S(X\times Y)\to SX$, 
naturally in $X$;
\item for every $X\in |\BC|$, there is $=_X\,\in S(X\times X)$ such that for 
all $\phi\in S(X\times X)$, $\top\leq (S\brks{\id_X,\id_X})(\phi)$ iff $=_X\,\leq \phi$.
\end{enumerate}
If additionally 
\begin{enumerate}[start=5]
\item given $X\in |\BC|$, $SX$ is a \emph{BI-algebra}, i.e.\ a commutative monoid equipped
with a right order-adjoint to multiplication;
\item given $f\in\BC(X,Y)$, $Sf\c SY\to SX$ is a BI-algebra morphism,
\end{enumerate}
then $S$ is called a \emph{first order BI-hyperdoctrine}.

In a \emph{(higher order) hyperdoctrine}, $\BC$ is additionally required to be Cartesian 
closed and every $SX$ is required to be poset-isomorphic to $\BC(X, A)$ for a 
suitable internal Heyting algebra $A\in |\BC|$ naturally in $X$. Such a hyperdoctrine is a 
\emph{BI-hyperdoctrine} if moreover~$A$ is an internal BI-algebra.
\end{definition}
\begin{proposition}
Every internally complete Heyting algebra $A$ in a Cartesian closed               
category $\BC$ with finite limits gives rise to a canonical hyperdoctrine       
$\BC(\argument, A)$: for every $X$, $\BC(X,A)$ is a poset under $f\leq g$ iff   
$f\land g = f$.                                                                 
\end{proposition}
\begin{proof}
Clearly, every $\BC(X,A)$ is a Heyting algebra and every $\BC(f,A)$ is a Heyting algebra 
morphism. The quantifies are defined mutually dually as follows:
\begin{align*}
 (\exists Y)_X (\phi\c X\times Y\to A) = \bigor_{\fst\c X\times Y\to X}\phi,\\
 (\forall Y)_X (\phi\c X\times Y\to A) = \bigand_{\fst\c X\times Y\to X}\phi.
\end{align*}
Naturality in $X$ follows from the corresponding Beck-Chevalley conditions. 

Finally, internal equality $=_X\c X\times X\to A$ is defined as $\bigor_{\brks{\id_X,\id_X}}\top$.
\qed\end{proof}

\begin{figure}[t!]
\begin{center}%
{
\parbox{\textwidth}{%
\small%
\begin{flalign*}
\anonrule{}{%
  \G\vctx v\c A \quad \G\ctx\phi\c\pred A 
  }{%
  \G\ctx\phi(v)\c\pro
  }
&&
\anonrule{}{%
  \G,x\c A\ctx\phi\c\pro 
  }{%
  \G\ctx x.\,\phi\c\pred A
  }
&&
\anonrule{}{%
  \G\vctx \ell\c\Ref_S \quad \G\vctx v\c \CType(S)
  }{%
  \G\ctx \ell\ito v\c\pro
  }
\end{flalign*}
\begin{flalign*}
&&
\anonrule{}{%
  \G\ctx\phi\c\pred A
  }{%
  \G\ctx Q\,\phi\c\pro
  }
\quad (Q\in\{\forall,\exists\})
&&
\anonrule{}{%
  \G\vctx v\c A \qquad \G\vctx w\c A
  }{%
  \G\ctx v=w\c\pro
  }
&&
\end{flalign*}
\begin{flalign*}
\anonrule{}{%
  }{%
  \G\ctx c\c\pro
  }
\quad (c\in\{\top,\bot\})
&&
\anonrule{}{%
  \G\ctx\phi\c\pro
    \qquad
  \G\ctx\psi\c\pro
  }{%
  \G\ctx\phi\,\$\,\psi\c\pro
  }
\quad (\$\in\{\land,\lor,\tto,\star,\sepimp\})
\end{flalign*}
}}
\end{center}
\caption{Term formation rules for the higher order separation logic.}
\label{fig:logic.lang}
\end{figure}

\noindent
A standard way to obtain an (internally) complete BI-algebra is to resort to 
ordered partial commutative monoids~\cite{PymOHearnEtAl04}.
\begin{definition}[Ordered PCM~\cite{PymOHearnEtAl04}]
An \emph{ordered partial commutative monoid (pcm)} is a tuple $(\CM, \CE, \cdot\,, \leq)$ 
where~$\CM$ is a set, $\CE\subseteq\CM$ is a set of \emph{units}, \emph{multiplication} $\cdot$
is a partial binary operation on $\CM$, and $\leq$ is a preorder on $\CM$, such that the following axioms 
are satisfied (where $n\simeq m$ denotes \emph{Kleene equality} of~$n$ and~$m$,
i.e.\ both $n$ and $m$ are defined and equal):
\begin{enumerate}
\item $m\cdot n\simeq m\cdot n$;
\item $(m\cdot n) \cdot k \simeq m\cdot (n\cdot k)$;
\item for any $m\in\CM$ there is $e \in\CE$ such that $m\cdot e \simeq m$;
\item for any $m\in\CM$ and any $e \in\CE$, if $m\cdot e$ is defined then $m\cdot e \simeq m$;
\item if $n'\leq n$, $m'\leq m$, and $n\cdot m$ is defined then so is $n'\cdot m'$ 
and $n'\cdot m'\leq n\cdot m$.
\end{enumerate}
\end{definition}
We note that using general recipes~\cite{BizjakBirkedal18}, for every internal 
ordered pcm $M$ in a topos $\BC$ with subobject classifier $\Omega$, 
$\BC(\argument\times M,\Omega)$ forms a BI-hyperdoctrine, on particular, if $\BC=\Set$
then $\Set(\argument\times M, 2)$ is a BI-hyperdoctrine.

\section{A Higher Order Logic for Full Ground Store}\label{sec:sep}
We proceed to develop a local version of separation logic using semantic
principles explored in the previous sections. That is, we seek an interpretation 
for the language in Fig.~\ref{fig:logic.lang} in the category 
$[\BW,\Set]$ over the type system~\eqref{eq:gramma}, extended with \emph{predicate 
types} $\pred A$. The judgements $\G\ctx\phi\c\pro$ type formulas depending 
on a variable context $\G$. Additionally, we have judgements of the form
$\G\ctx\phi\c\pred A$ for \emph{predicates in context}. Both kinds of 
judgements are mutually convertible using the standard application-abstraction 
routine. Note that expressions for quantifiers $\exists x.\,\phi$ are thus obtained
in two steps: by forming a predicate $x.\,\phi$, and subsequently applying $\exists$.
Apart from the standard logical connectives, we postulate \emph{separating conjunction} 
$\star$ and \emph{separating implication}~$\sepimp$. 

Our goal is to build a BI-hyperdoctrine, using the recipes, summarized in the 
previous section. That is, we construct a certain internal BI-algebra $\TA$ in $[\BW,\Set]$,
and subsequently conclude that $[\argument,\TA]$ is a BI-hyperdoctrine in question. In what 
follows, most of the effort is invested into constructing an internally complete 
Boolean algebra $\check\PSet\comp\hat P\hat H$ (hence $[\argument,\check\PSet\comp\hat P\hat H]$
is a hyperdoctrine), from which $\TA$ is carved out as a subfunctor, 
identified by an upward closure condition. Here, $\check\PSet$ is a contravariant
powerset functor, and~$\hat P$ and~$\hat H$ are certain modifications of the 
hiding and the heap functors from Section~\ref{sec:fgs}. As we shall see, the move from 
$\check\PSet\comp\hat P\hat H$ to $\TA$ remedies the problem of the former 
that the natural separation conjunction operator $\star$ on it does not have unit 
(Remark~\ref{rem:nempty}).

In order to model resource separation, we must identify a domain of logical 
assertions over partial heaps, i.e.\ heaplets, instead of total heaps.
We thus need to derive a unary (covariant) heaplet functor from the binary,
mix-variant one $\CH$ used before. We must still cope not 
only with heaplets, but with partially hidden heaplets, to model information hiding.
A seemingly natural candidate functor for hidden heaplets is the composition
\begin{displaymath}
  P\bigl(\BE
  \xto{\sum_{w\subseteq\argument} \Hplb (w, \argument)} 
\Set\bigr)\c\BW^\op\to\Set. 
\end{displaymath}
One problem of this definition is that the equivalence relation 
$\sim$ underlying the construction of $P$ in~\eqref{eq:hid} is too fine. Consider, for example, 
$e_w = (\emptyset\subseteq w, \star)\in\sum_{w'\subseteq w}\CH(w',w)$. 
Then $(\id\c w\to w, e_w)\nsim (\inl\c w \to w \oplus \{\star\c 1\}\comma e_{w \oplus\{\star\c 1\}})$,
i.e.\ two hidden heaplets would not be equivalent if one extends the other by 
an inaccessible hidden cell. In order to arrive at a more reasonable model of logical 
assertions, we modify the previous model by replacing the category 
of initializations~$\BE$ is a category $\hat\BE$ of \emph{partial initializations}.
This will induce a hiding monad $\hat P$ over $[\hat\BE,\Set]$ using exactly the 
same formula~\eqref{eq:hid} as for $P$.

A partial initialization is a pair
$(\rho, \eta)$ with $\rho\in\BW(w^-_1, w^+_2)$ and
$\eta \in\sum_{w^- \subseteq w^+_2 \ominus \rho} \Hplb(w^-, w^+_2)$.
Let $\hat\BE$ be the category of heap layouts and partial initializations.
Analogously to~$\geoU$, there is an obvious partial-heap-forgetting functor $\hat\geoU\c\hat\BE \to\BW$.
Let $\hat H\c\hat\BE \to\Set $ be the following \emph{heaplet functor}:
\begin{align*}
	\hat H w =&\; \sum_{w'\subseteq w} \Hplb (w', w).
\end{align*}
Given a partial initialization $\epsilon = (\rho\c w \to w',(w''\subseteq w'\ominus \rho, \eta
\in\Hplb(w'', w')))\c w\into w'$, $\hat H\epsilon\c\hat H w\to\hat Hw'$ 
extends a given heaplet over $w$ to a heaplet over~$w'$ via $\eta$:
\begin{align*}
 (\hat H \epsilon)(w_1\subseteq w,\eta'\in\Hplb(w_1,w)) =&\; (\rho [w_1]\cup w''\subseteq w',\eta'')
\end{align*}
where $\eta''\in\Hplb(\rho[w_1]\cup w''\subseteq w',w')$ is as follows
\begin{flalign*}
&&\pr_{\rho(\ell\c S)}\eta'' =\;& \oname{range}(S)(\rho)(\pr_{(\ell\c S)}\eta') & ((\ell\c S)\in w_1)\\
&&\pr_{(\ell\c S)}\eta'' =\;& \pr_{(\ell\c S)}\eta & ((\ell\c S)\in w'')
\end{flalign*}
With $\hat\BE$ and $\hat H$ as above instead of $\BE$ and $H$, the framework 
described in Section~\ref{sec:fgs} transforms coherently.
\begin{remark}
Let us fix a fresh symbol $\inacc$, and note that 
\begin{align*}
\hat H w =\sum_{w'\subseteq w}\prod_{(\ell\c S)\in w'} \oname{range}(S)(w)
\iso \prod_{(\ell\c S)\in w} (\oname{range}(S)(w)\uplus\{\inacc\}),
\end{align*}
meaning that the passage from $\BE$, $H$ and $P$ to $\hat\BE$, $\hat H$ and 
$\hat P$ is equivalent to extending the $\oname{range}$ function with 
designated values $\inacc$ for \emph{inaccessible locations}. 
We prefer to think of $\inacc$ this way and not as a content of \emph{dangling pointers}, 
to emphasize that we deal with a \emph{reasoning phenomenon} and not with a 
\emph{programming phenomenon}, for our programs neither create nor process 
dangling pointers.
\end{remark}

For the next proposition we need the following concrete description of the set 
$\hat\geoU_\star(2^X)w$ as the end $\int_{\rho\c w\to w'\in w \dar\hat\geoU} \Set(X w',2)$:
this set is a space of dependent functions $\phi$ sending every injection $\rho\c w\to w'$ 
to a corresponding subset of~$Xw'$, and satisfying the constraint: 
$x \in \phi (\rho)$ iff $(X \eps)(x)\in \phi(\hat\geoU\eps \comp \rho)$ for every
${\eps\c w'\into w''}$.
\begin{proposition}\label{prop:p-exp}
\needspace{3\baselineskip}
The following diagram commutes up to isomorphism:
\begin{center}
\begin{tikzcd}[column sep=huge]
{[\hat\BE,\Set]} 
\rar["2^{(\argument)}"] 
\dar["\hat P"'] 
& 
{[\hat\BE,\Set]^\op} 
\dar["\hat\geoU_\star"]
\\
{[\BW,\Set^\op]^\op} 
\rar["{\check\PSet\comp\, {(\argument)}}"]
& 
{[\BW,\Set]^\op}
\end{tikzcd}
\end{center}
(using the fact that $[\BW,\Set^\op]^\op\iso[\BW^\op,\Set]$) where $\check\PSet$ is 
the contravariant powerset functor $\check\PSet\c\Set^\op\to\Set$ and for every 
$X\c\hat\BE\to\Set$ the relevant isomorphism 
$\Phi_w\c\hat\geoU_\star (2^X)w\iso\check\PSet (\hat PX w)$ is as follows:
\begin{align}\label{eq:Phi-iso}
&(\rho\c w \to w', x \in X w')_\sim\in\Phi_w (\phi\in\hat\geoU_\star (2^X) w)\iff	x\in\phi(\rho).
\end{align}
\end{proposition}
Let us clarify the significance of Proposition~\ref{prop:p-exp}. The exponential 
$2^{\hat H}$ in $[\hat\BE,\Set]$ can be thought of as a carrier of Boolean predicates
over $\hat H$, and as we see next
those form an internally complete Boolean algebra, which is carried from 
$[\hat\BE,\Set]$ to $[\BW,\Set]$ by $\hat\geoU_\star$. The alternative route
via $\hat P$ and $\check \PSet$ induces a Boolean algebra of predicates 
over hidden heaplets $\hat P\hat H$ directly in $[\BW,\Set]$. The equivalence 
established in Proposition~\ref{prop:p-exp} witnesses agreement of these two structures. 

\begin{theorem}\label{thm:int-BA}
For every $X\c\hat\BE\to\Set$, $\check\PSet\comp \hat PX$ is an internally 
complete Boolean algebra in $[\BW,\Set]$ under
\begin{align*}
\Bigl(\bigor_f\phi\c I& \to\check\PSet\comp \hat P X\Bigr)_w (j \in J w)\\*
 =&\, \{ (\rho\c w\to w', x \in Xw')_\sim\mid
\exists\eps\c w	'\into w'', \exists i \in I w''.\,\\*
&\qquad f_{w''}(i) = J(\hat\geoU\eps\comp \rho)(j)\land (\id_{w''}, (X\eps)(x))_\sim\in 
\phi_{w''}(i)\},\\
\Bigl(\bigand_f\phi\c I& \to\check\PSet\comp \hat P X\Bigr)_w (j \in J w)\\
 =&\, \{ (\rho\c w\to w', x \in Xw')_\sim\mid
\forall\eps\c w'\into w'', \forall i \in I w''.\,\\
&\qquad f_{w''}(i) = J(\hat\geoU\eps\comp \rho)(j)\impl (\id_{w''}, (X\eps)(x))_\sim\in 
\phi_{w''}(i)\}.
\end{align*}
for every $f \c I \to J$, and the corresponding Boolean algebra operations are 
computed as set-theoretic unions, intersections and complements.

\end{theorem}
By Theorem~\ref{thm:int-BA}, we obtain a hyperdoctrine $[\argument, \check\PSet\comp (\hat P\hat H)]$, 
which provides us with a model of (classical) higher order logic in $[\BW,\Set]$.
In particular, this allows us to interpret the language from
Fig.~\ref{fig:logic.lang} over $[\BW,\Set]$ excluding the separation logic constructs, 
in such a way that 
\begin{align*}
\sem{\G\ctx\phi\c\pro}\c\ul{\G}\to\check\PSet\comp (\hat P\hat H),&&
\sem{\G\ctx\phi\c\pred A}\c\ul{\G}\times\ul{A}\to\check\PSet\comp (\hat P\hat H)
\end{align*}
where $\ul{\G}=\ul{A_1}\times\ldots\times\ul{A_n}$ for $\G=(x_1\c A_1,\ldots,x_n\c A_n)$
where, additionally to the standard clauses, $\ul{\pred A} = \check\PSet\comp\hat P(\geoU^\star \ul{A}\times\hat H)$.
The latter interpretation of predicate types $\pred A$ is justified by the natural isomorphism:
\begin{align*}
(\check\PSet\comp (\hat P\hat H))^X 
\iso (\hat\geoU_\star(2^{\hat H}))^X
\iso \hat\geoU_\star((2^{\hat H})^{\hat\geoU^\star X})
\iso \check\PSet\comp (\hat P(\hat\geoU^\star X\times\hat H)).
\end{align*}
Here, the first and the last transitions are by $\Phi$ from
Proposition~\ref{prop:p-exp} and the middle one is due to the fact that clearly
both $(\hat\geoU_\star(\argument))^X\vdash\hat\geoU^\star (X\times(\argument))$
and $\hat\geoU_\star((\argument)^{\hat\geoU^\star X})\vdash\hat\geoU^\star
(X\times(\argument))$. %

Since every set $\hat Hw$ models a heaplet in the standard sense~\cite{PymOHearnEtAl04}, we can 
equip~$\hat Hw$ with a standard pointer model structure. 
\begin{proposition}
For every $w\in |\BW|$, $(\hat H w,\{(\emptyset\subseteq w,\star)\},\cdot\,,\leq)$ 
is an ordered pcm where for every $w\in |\BW|$, $\hat Hw$ is partially ordered as follows:
\begin{flalign*}
&&(w_1\subseteq w,\Hplb(w_1\subseteq w_2,w)\eta\in\Hplb(w_1,w))\leq (w_2\subseteq w\comma\eta\in\Hplb(w_2,w))&&
(w_1\subseteq w_2)
\end{flalign*}
and for $w_1\subseteq w$, $w_2\subseteq w$ and $\eta_1\in\Hplb(w_1,w)$, 
$\eta_2\in\Hplb(w_2,w)$, $(w_1\subseteq w,\eta_1)\cdot (w_2\subseteq w,\eta_2)$ 
equals $(w_1\cup w_2,\eta_1\cup\eta_2)$ if $w_1\cap w_2=\emptyset$, and otherwise
undefined.
\end{proposition}
As indicated in Section~\ref{sec:bi}, we automatically obtain a BI-algebra 
structure over the set of all subsets of $\hat H w$. The same strategy does not
apply to $\hat P\hat H w$, roughly because we cannot predict mutual arrangement 
of hidden partitions of two heaplets wrt to each other, for we do not have a global
reference space for pointers as contrasted to the standard separation logic setting.
We thus define a separating conjunction operator directly on every 
$\check\PSet (\hat P \hat H w)$ as follows:
\begin{align*}
	\phi \star_w \psi = \{(\rho\c w \to w',&\; (w_1 \uplus w_2 \subseteq w', 
		\eta \in\Hplb(w_1 \uplus w_2, w')))_\sim\mid\\*  
		&\qquad(\rho, (w_1 \subseteq w', \Hplb(w_1 \subseteq w_1 \uplus w_2, w')\eta))_\sim
		\in\phi, \\*
		&\qquad (\rho, (w_2 \subseteq w', \Hplb(w_2 \subseteq w_1 \uplus w_2, w')\eta))_\sim
		\in\psi\}.
\end{align*}

\begin{lemma}\label{lem:sep-props}
The operator $\star_w$ on $\check\PSet (\hat P \hat H w)$
satisfies the following properties.
	\begin{enumerate}
		\item\label{it:sep-prop1} $\star_w$ is natural in $w$.
		\item\label{it:sep-prop2} $\star_w$ is associative and commutative.
		\item\label{it:sep-prop3} $(\rho\c w \to w', (w''\subseteq w', 
		\eta \in\Hplb(w'', w')))_\sim\in\phi \star_w \psi$ if and only if
		there exist~$w_1$ and~$w_2$ such that $w_1 \uplus w_2 = w''$,
		$(\rho, (w_1 \subseteq w', \Hplb(w_1 \subseteq w'', w')\eta))_\sim
		\in\phi$ and
		$(\rho, (w_2 \subseteq w', \Hplb(w_2 \subseteq w'', w')\eta))_\sim
		\in\psi$.		
	\end{enumerate}
\end{lemma}
\noindent
Property~\textit{\ref{it:sep-prop3}}. specifically tells us that any representative of an equivalence 
class contained in a separating conjunction can be split in such a way that the 
respective pieces belong to the arguments of the separating conjunction.

\begin{remark}\label{rem:nempty}
The only candidate for the unit of the separating conjunction $\star_w$ would be the emptiness predicate
$\oname{empty}_w\c 1\to\check\PSet (\hat P \hat H w)$, identifying precisely
the empty heaplets. However, $\oname{empty}_w$ is not natural in $w$. In fact,
it follows by Yoneda lemma that there are exactly two natural transformations $1
\to\check\PSet\comp\hat P\hat H$, which are the total truth and the total
false, none of which is a unit for $\star_w$.
\end{remark}
Remark~\ref{rem:nempty} provides a formal argument why we cannot interpret 
classical separation logic over $\check\PSet\comp\hat P\hat H$. We thus 
proceed to identify for every $w$ a subset of $\check\PSet (\hat P \hat Hw)$, for which the total truth predicate becomes 
the unit of the separating conjunction. Concretely, let $\TA$ be the subfunctor
of $\check\PSet\comp\hat P\hat H$ identified by the following \emph{upward 
closure condition:} $\phi\in\TA w$ if
\begin{align*}
(\rho,\eta)_\sim\in\phi,~ \eta \leq \eta'\text{\qquad imply\qquad} (\rho,\eta')_\sim\in\phi.
\end{align*}
For every $w$, let $\ucl_w\c\PSet(\hat P \hat H w)\to\TA w$ send a set 
$\phi\subseteq\hat P \hat H w$ to the smallest  upward closed subset of $\hat P \hat H w$ 
containing $\phi$.
\begin{lemma}\label{lem:B-cHA}
$\TA$ is an internal complete sublattice of $\check\PSet\comp\hat P\hat H$, i.e.\ the inclusion 
$\iota\c\TA\ito\check\PSet\comp\hat P\hat H$ preserves all meets and all joins. 
This canonically equips $\TA$ with an internally complete Heyting algebra structure.
\end{lemma}
\begin{proof}[Sketch]
The key idea is to use the retraction $(\iota,\ucl)$.
The requisite structure is then transferred from $\check\PSet\comp\hat P\hat H$ 
to $\TA$ along it. The Heyting implication for $\TA$ is obtained using the 
standard formula ${(\phi\To\psi)} =\bigor\{\xi\mid \phi\land\xi\le\psi\}$
interpreted in the internal language.
\qed\end{proof}

\begin{lemma}\label{lem:cl-preserve}
Separating conjunction preserves upward closure:
for $\phi, \psi \in\TA w$, $\phi \star_w \psi = \ucl_w (\phi \star_w \psi)$.
\end{lemma}

\begin{lemma}\label{lem:B-BI}
$\TA$ is a BI-algebra: $\star_w$ is obtained by restriction from $\check\PSet(\hat P\hat Hw)$ by 
Lemma~\ref{lem:cl-preserve}, $\hat P\hat Hw$ is the unit for it and
\begin{align*}
\phi \sepimp_w \psi =&\; \{(\rho,\eta)_\sim\in\TA w\mid \forall \rho'\c w \to w', 
\eta_1, \eta_2 \in\hat H w', 
\eta_1 \cdot\eta_2 \text{ defined } \land\\&\qquad\qquad\quad
(\rho, \eta)\sim (\rho',\eta_1)\land (\rho', \eta_2)_\sim\in\phi\impl
(\rho', \eta_1\cdot\eta_2)_\sim\in\psi \}.
\end{align*}
\end{lemma}
\begin{proof}
In view of Lemma~\ref{lem:B-cHA}, we are left to show that the given operations
are natural and that $\TA$ is an internal BI-algebra w.r.t.\ them. Since BI-algebras 
form a variety~\cite{GalatosJipsenEtAl07}, it suffices to show that each $\TA w$ is a BI-algebra. By 
Lemma~\ref{lem:sep-props}~(ii), it suffices to show that every 
$(\argument)\star_w\phi$ preserves arbitrary joins, for then we can use the 
standard formula to calculate $\phi\sepimp_w \psi$, which happens to be
natural in~$w$:
\begin{displaymath}
\phi\sepimp_w \psi = \bigcup\, \{\xi\mid \phi\star_w\xi\leq \psi \}.  
\end{displaymath}
By unfolding the right-hand side, we obtain the expression for $\sepimp_w$ figuring 
in the statement of the lemma.
\qed\end{proof}
\begin{figure}[t!]
\begin{itemize}[itemsep=1ex]
\item $s,\rho,\eta\models\top$ %
\item $s,\rho,\eta\models\phi\land\psi$ ~~if~~ $s,\rho,\eta\models\phi$ and $s,\rho,\eta\models\psi$
\item $s,\rho,\eta\models\phi\lor\psi$ ~~if~~ $s,\rho,\eta\models\phi$ or $s,\rho,\eta\models\psi$
\item $s,\rho,\eta\models\phi\To\psi$ ~~if~~ for all $(\rho,\eta)\sim (\rho',\eta')$ and
 $\eta'\leq\eta''$,\\[1ex]\erule\quad  $s,\rho',\eta''\models\phi$ implies $s,\rho',\eta''\models\psi$
\item $s,\rho,\eta\models\phi(v)$ ~~if~~ $s,\rho,((\sem{\G\vctx v\c A}_{w'}\comp\ul\Gamma\rho)s,\eta)\models\phi$ 
\item $s,\rho,(a,\eta)\models x.\,\phi$ ~~if~~ $a=(X\rho) b$ and $(s,b),\rho,\eta\models\phi$ 
\item $s,\rho,\eta\models  \ell\ito v$ ~~if~~  $\eta=(w''\subseteq w',\delta\in\CH(w'',w'))$ 
and\\[1ex]\erule\quad $\delta(r\c S) = (\sem{\G\vctx v\c \CType(S)}_{w'}\comp\ul{\G}\rho)s$\\[1ex]\erule\quad where 
$(\sem{\G\vctx \ell\c\Ref_S}_{w'}\comp\ul{\G}\rho)s = (r\c S)\in w''$
\item $s,\rho,\eta\models v = u$ ~~if~~
  $(\sem{\G\vctx v\c A}_{w''}\comp\ul{\G}\rho'\comp\ul{\G}\rho)(s) = (\sem{\G\vctx u\c A}_{w''}\comp\ul{\G}\rho'\comp\ul{\G}\rho)(s)$
  \\[1ex]\erule\quad for some $\rho'\c w'\to w''$
\item $s,\rho,\eta\models\phi\star\psi$ ~~if~~ for suitable $w_1$, $w_2$, 
$\eta\in\Hplb(w_1 \uplus w_2, w')$,\\[1ex]\erule\quad $s,\rho,(w_1 \subseteq w', \Hplb(w_1 \subseteq w_1 \uplus w_2, w')\eta)\models\phi$
and\\[1ex]\erule\quad $s,\rho,(w_2 \subseteq w', \Hplb(w_2 \subseteq w_1 \uplus w_2, w')\eta)\models\psi$
\item $s,\rho,\eta\models\phi\sepimp\psi$ ~~if~~ for all $(\rho',\eta_1)\sim (\rho,\eta)$
and for all $\eta_2$ such that $\eta_1\cdot\eta_2$ is defined,\\[1ex]\erule\quad 
$s,\rho',\eta_2\models\phi$ implies $s,\rho',\eta_1\cdot\eta_2\models\psi$
\item $s,\rho,\eta\models\exists\phi$ ~~if~~ $\ul{\G}(\hat\geoU\epsilon\comp\rho)s,\id_{w''},(a,\hat H\epsilon\comp\eta)\models\phi$ for some 
$\epsilon\c w'\into w''$, $a\in\ul{A} w''$
\item $s,\rho,\eta\models\forall\phi$ ~~if~~ $\ul{\G}(\hat\geoU\epsilon\comp\rho)s,\id_{w''},(a,\hat H\epsilon\comp\eta)\models\phi$ for all 
$\epsilon\c w'\into w''$, $a\in\ul{A} w''$
\end{itemize}
\caption{Semantics of the logic.}
\label{fig:log.sem}
\end{figure}
\needspace{2\baselineskip}
\begin{theorem}
$\TA$ is an internally complete Heyting BI-algebra, hence $[\argument,\TA]$ is a 
$BI$-hyperdoctrine. 
\end{theorem}
\begin{proof}
Follows from Lemmas~\ref{lem:B-cHA} and~\ref{lem:B-BI}.\qed
\end{proof}
This now provides us with a complete semantics of the language in Fig.~\ref{fig:logic.lang}
with $\sem{\G\vdash\phi\c\pro}\c\ul{\G}\to\TA$ and
$\sem{\G\vdash\phi\c\pred{A}}\c\ul{\G}\to\ul{\pred A}$ where
$\ul{\pred A}$ is the upward closed subfunctor of $\check\PSet\comp (\hat P (\hat\geoU \ul{A}\times \hat H))$, with upward closure only on the $\hat H$-part, which
is isomorphic to $\TA^{\ul{A}}$.
The resulting semantics is defined in Fig.~\ref{fig:log.sem} where we write 
$s,\rho,\eta\models\phi$ for $(\rho,\eta)_\sim\in\sem{\G\vdash\phi\c\pro}(s)$
and $s,\rho,(a,\eta)\models\phi$ for $(\rho,(a,\eta))_\sim\in\sem{\G\vdash\phi\c\pred A}(s)$.
The following properties~\cite{CalcagnoOHearnEtAl03} are then automatic.
\begin{proposition}
\begin{itemize}
	\item \emph{(Monotonicity)} If $s,\rho,\eta\models\phi$  and $\eta\leq\eta'$ then
		$s, \rho, \eta' \models \phi$.
	\item \emph{(Shrinkage)} If $s,\rho,\eta\models\phi$, $\eta'\leq\eta$ and $\eta'$ contains all cells reachable from $s$ and $w$ then
	$s, \rho, \eta' \models \phi$.
		\end{itemize}
\end{proposition}

\section{Examples}\label{sec:exam}
Let us illustrate subtle features of our semantics by some examples.
\begin{example}
Consider the formula $\exists \ell\c \Ref_{\oname{Int}}. \, \ell\ito 5$  
from the introduction in the empty context $\argument$. Then $\argument,\rho,\eta\models\exists\ell. \, \ell\ito 5$
iff for some $\epsilon\c w'\into w''$, and some $x\in\ul{\Ref_{\oname{Int}}}w''$,
$x,\id_{w''}, (\hat H \epsilon)\eta\models\ell'\ito 5$. The latter is true iff 
$\pr_x ((\hat H \epsilon)\eta) = 5$. Note that $w'$ may not contain $\ell$ and it 
is always possible to choose $\eps$ so that~$w''$ contains $\ell$ and 
$\pr_x ((\hat H \epsilon)\eta) = 5$. Hence, the original formula is always~valid.
\end{example}

\begin{example}\label{expl:impl} The clauses in Fig.~\ref{fig:log.sem} are very similar to the 
standard Kripke semantics of intuitionistic logic. Note however, that the clause for implication 
strikingly differs from the expected one
\begin{itemize}
 	\item $s,\rho,\eta\models\phi\To\psi$ ~~if~~ for all
 $\eta\leq\eta'$, $s,\rho,\eta'\models\phi$ implies $s,\rho,\eta'\models\psi$,
\end{itemize}
though. The latter is indeed not validated by our semantics, as witnessed by the following
example. Consider the following formulas $\phi$ and $\psi$ respectively:
 \begin{align}
 	\ell\c \Ref_{\Ref_{\mathit{Int}}}\,&\ctx\exists \ell'.\,\exists x.\, \ell\ito \ell' \land \ell' \ito x\c\pro\label{eq:imp-ex1}\\*
 	\ell\c \Ref_{\Ref_{\mathit{Int}}}\,&\ctx\exists \ell'.\, \ell\ito \ell' \land \ell' \ito 6\c\pro\label{eq:imp-ex2}
 \end{align}
The first formula is valid over heaplets, in which $\ell$ refers to a reference to 
some integer, while the second one is only valid over heaplets, in which $\ell$ refers to a 
reference to~$6$.
Any $\eta'\geq\eta = (\id_w, (\{\ell''\}\subseteq\{\ell,\ell''\}, [\ell'' \mto 6]))$ satisfies both~\eqref{eq:imp-ex1}
and~\eqref{eq:imp-ex2} or none of them. However, the implication $\phi\To\psi$ 
still is not valid over~$\eta$ in our semantics, for
\begin{align*}
\eta\sim&\; 
(w \ito w \oplus {(\ell'\c  \mathit{Int})}, (\{\ell',\ell''\}\subseteq\{\ell,\ell',\ell''\}, [\ell' \mto 5, \ell'' \mto 6]))\\
 	\leq &\; (w \ito w \oplus {(\ell'\c  \mathit{Int})}, (\{\ell,\ell',\ell''\}\subseteq\{\ell,\ell',\ell''\}, [\ell\mto \ell', \ell' \mto 5, \ell'' \mto 6]))
\end{align*}
and the latter heaplet validates $\phi$ but not $\psi$.
\end{example}

\begin{example}
Least $\mu$ and greatest $\nu$ fixpoints can be encoded in higher order logic~\cite{BieringBirkedalEtAl07}.
As an example, consider 
\begin{align*}
\mathit{isList} = \mu\gamma.\,\ell.\,\,\ell \ito\mathit{null} \lor \exists \ell',x.\, \ell\ito (x, \ell')\star \gamma(\ell'),
\end{align*}
which specifies the fact that $\ell$ is a pointer to a head of a list (eliding 
coproduct injections in $\inl\mathit{null}$ and $\inr(x, \ell')$). By definition,
$\mathit{isList}$ satisfies the following recursive equation: 
\begin{align*}
\mathit{isList}(\ell) = \ell\ito\mathit{null} \lor \exists \ell',x.\, \ell\ito (x, \ell')\star \mathit{isList}(\ell')
\end{align*}
Let us expand the semantics of the right hand side. We have
\begin{align*}
\sem{\ell\c&\Ref_{\oname{list}}, \mathit{isList}\c\pred{(\Ref_{\oname{list}})} \vdash
l \ito\mathit{null} \lor \exists \ell',x.\, \ell\ito (x, \ell')\star \mathit{isList}(\ell')}_w ( \mathit{isList})\\
 =&\;\{(\rho\c w\to w',(\ul{\Ref_{\oname{list}}}\rho)(\ell), \delta \in\hat H w')_\sim\mid \pr_{\rho(\ell)}(\delta) = \mathit{null}
\} \cup\\
&\qquad\sem{\ell\c\Ref_{\oname{list}}, \mathit{isList}\c\pred{(\Ref_{\oname{list}})} \vdash \exists \ell',x.\, \ell\ito (x, \ell')\star \mathit{isList}(\ell')}_w (\mathit{isList})\\
 =&\;\{(\rho\c w\to w',(\ul{\Ref_{\oname{list}}}\rho)(\ell), \delta \in\hat H w')_\sim\mid\\
&\qquad\pr_{\rho(\ell)}(\delta) = \mathit{null}~\lor~ \exists \ell',~ x.\,\pr_{\rho(\ell)} \delta = (x, \ell')\land (\rho, \ell',\delta \smin\rho(\ell))_\sim
\in\mathit{isList}\}
\end{align*}
where $\delta \smin\rho(\ell)$ denotes the $\delta$ with the cell $\rho(\ell)$
removed. In summary, $(\rho\c w\to w',(\ul{\Ref_{\oname{list}}}\rho)(\ell), \delta \in\hat H w')_\sim$ is in 
$\sem{\ell\c\Ref_{\oname{list}}, \mathit{isList}\c\pred{(\Ref_{\oname{list}})} \vdash
\mathit{isList}(\ell)}_w (\mathit{isList})$ if and only if either $\pr_{\rho(\ell)}\delta = \mathit{null}$
or there exists an $l'\in w'$ such that $\pr_{\rho(\ell)} \delta = (x, \ell')$
and $(\rho, \ell', \delta \smin\rho(\ell))_\sim
	\in\mathit{isList}$.
\end{example}

\section{Conclusions and Further Work}\label{sec:conc}
Compositionality is an uncontroversial desirable property in semantics and 
reasoning, which admits strikingly different, but equally valid interpretations,
as becomes particularly instructive when modelling dynamic memory allocation. 
From the programming perspective it is desirable to provide compositional
means for keeping track of integrity of the underlying data, in particular, for preventing 
\emph{dangling pointers}. Reasoning however inherently requires introduction of partially 
defined data, such as \emph{heaplets}, which due to the compositionality principle 
must be regarded as first class semantic units. 

Here we have made a step towards reconciling recent extensional monad-based
denotational semantic for full-ground store~\cite{KammarLevyEtAl17} with 
higher order categorical reasoning frameworks~\cite{BieringBirkedalEtAl07} by
constructing a suitable intuitionistic BI-hyperdoctrine. Much remains to be
done. A highly desirable ingredient, which is currently missing in our logic
in Fig.~\ref{fig:logic.lang} is a construct relating programs and logical
assertions, such as the following dynamic logic style modality
\begin{align*}
\anonrule{}{%
  \G\cctx p\c A\qquad\G\ctx\phi\c\pred A   
  }{%
  \G\ctx [p]\phi\c\pro
  }
\end{align*}
which would allow us e.g.\ in a standard way to encode \emph{Hoare triples} 
$\htri{\phi}{p}{\psi}$ as implications $\phi\impl [p]\psi$. 
This is difficult 
due to the outlined discrepancy in the semantics for construction and reasoning. 
The categories of initializations for~$p$ and $\phi$ and the corresponding hiding 
monads are technically incompatible. In future work we aim to deeply analyse this
phenomenon and develop a semantics for such modalities in a principled fashion.

Orthogonally to these plans we are interested in further study of the full ground store monad
and its variants. One interesting research direction is developing algebraic 
presentations of these monads in terms of operations and 
equations~\cite{PlotkinPower03}. Certain generic methods~\cite{MaillardMellies15} 
were proposed for the simple store case (Example~\ref{def:gs}), and it remains to 
be seen if these can be generalized to the full ground store~case.

\clearpage

\bibliographystyle{plain}
\bibliography{bi-hyper-refs}

\clearpage
\appendix
\allowdisplaybreaks

\section{Appendix: Omitted Details}

\subsection{Proof of Lemma~\ref{prop:diamond}}
Assuming that $(\rho,x)\preceq (\rho_1,x_1)$ is witnessed by some $\eps_1\c w\into w_1$ 
and $(\rho,x)\preceq (\rho_2,x_2)$ is witnessed by some $\eps_2\c w\into w_2$, 
$(\rho_1,x_1)\preceq (\rho',x')$ and $(\rho_2,x_2)\preceq (\rho',x')$ are witnessed 
by the induced injections $w_1\into w_1\oplus_w w_2$ and  
$w_2\into w_1\oplus_w w_2$ respectively.
\qed

\subsection{Proof of Theorem~\ref{thm:simp-store}}
We will need two lemmas.
\begin{lemma}\label{lem:simpleStore1}
In the simple store model, for each $X\c\BE\to\Set$, $\geoU_\star (X^{H})\iso X(\argument)^{\CV^{(\argument)}}\c\BW\to\Set$,
	which
	 is equipped with the following functorial action:
	\begin{align*}
		(X(\argument)^{\CV^{(\argument)}}) (\rho \c w\to w') (p\c  {\CV^w}\to Xw) (s\c w '\to\CV) =
			X (\rho, s\comp\rho^\complement) (p(s\comp\rho))
	\end{align*}
	(Note that $(\rho, s\comp\rho^\complement)\in\BE(w,w')$,  hence $X(\rho, s\comp\rho^\complement)\c Xw\to Xw'$).
\end{lemma}
\begin{proof}
Let $X \c\BE \to\Set$. We will show that $X(\argument)^{\CV^{(\argument)}}$ is the right Kan extension
of~$X^H$ along $\geoU$, i.e.\ $X(\argument)^{\CV^{(\argument)}} = \geoU_\star (X^H)$. By definition,

\begin{align*}
(\geoU_\star X^H)w 
\iso\int_{\rho\c w\to w'\in w \dar \geoU} \Set(H w',Xw')
 = \int_{\rho\c w\to w'\in w \dar \geoU} \Set(\CV^{w'},Xw').
\end{align*}
The resulting end is the subset of $\prod_{\rho\c w\to w'} \Set(\CV^{w'},Xw')$,
consisting of those dependent maps $f$, which satisfy 
\begin{align*}
  (X\eps)\comp f(\rho)\comp\CV^{\geoU\eps} = f(\geoU\eps\comp\rho)
\end{align*}
for all $\rho\c w\to w'$ and ${\eps\c w'\into w''}$. Every such map is determined
by its action on $\id_w\c w\to w$, because for every $\rho\c w\to w'$, given 
$s\c w'\to\CV$, $f(\rho)(s) = f(\geoU\eps\comp\id_w)(s) = (X\eps\comp f(\id_w))(s\comp\rho)$ where 
$\eps = (\rho,s\comp\rho^\complement)\in\BE(w,w')$. Hence, $(\geoU_\star X^H)w$ 
is indeed isomorphic to $(Xw)^{\CV^w}$. The functorial action of our end is induced 
by the functorial action of $\prod_{\rho\c(\argument)\to w'} \Set(\CV^{w'},Xw')$
sending every $\tau\c w_1\to w_2$ to $f\mto f(\argument\comp\tau)$. By composing 
this functorial action with the isomorphism $(\geoU_\star X^H)w\iso (Xw)^{\CV^w}$, we obtain that 
\begin{align*}
	(X(\argument)^{\CV^{(\argument)}}) (\tau \c w_1\to w_2) (p\c\CV^{w_1}\to Xw_1)(s\c w_2\to\CV)
	= (X(\tau,s\comp\tau^\complement)\comp p)(s\comp\tau),
\end{align*}
which is equivalent to the goal.
\qed\end{proof}
Next, we observe that for the simple store, the coend over a covariant functor
figuring in the definition of~$T$ can be transformed to a coend over a properly mix-variant
functor, thus, the whole category $\BE$ can be eliminated from the game.
\begin{lemma}\label{lem:simple-store2}
In the simple store model, $\Hid(X\comp\geoU\times H)w$ is equivalently the quotient
of\/ $\sum_{\rho\c w \to w'} X w'\times \CV^{w'}$ under the equivalence
relation $\sim$, generated by the clauses:
\begin{align*}
(\rho, (x, s\comp\rho'))\sim (\rho'\comp\rho, ((X\rho')(x), s))
\end{align*}
with $\rho\c w\to w'$, $\rho'\c w'\to w''$, $s\c w''\to\CV$ and $x \in X w'$.
Hence,
\begin{align*}
\Hid(X\comp\geoU\times H)w\iso\int^{\rho\c w\to w'\in w \dar\BW} Xw'\times \CV^{w'}. 
\end{align*} 
\end{lemma}
The proof of Theorem~\ref{thm:simp-store} is now obtained as follows.

\begin{flalign*}
&&\geoU_\star(\Hid(\geoU^\star X\times H)^H)w  &\;\iso (\Hid(X\comp\geoU\times H)w)^{\CV w} &\by{Lemma~\ref{lem:simpleStore1}}\\
&&  &\;\iso \biggl(\int^{\rho\c w\to w'\in w \dar \BW} Xw'\times \CV^{w'}\biggr)^{\CV^w}.&\by{Lemma~\ref{lem:simple-store2}}
\end{flalign*}
\qed
\begin{figure}[t!]
\begin{center}%
{\parbox{\textwidth}{%
\small%
\begin{flalign*}
&&
\anonrule{get}{%
  \sem{\G\vctx \ell\c\Ref_S} = \alpha\c\ul{\G}\to\sem{\Ref_S}
  }{%
  \sem{\G\cctx\bang \ell\c \CType(S)} = \geoU_\star(\Hid\oname{put})^H\comp\eta\comp\alpha\c\ul{\G}\to T\ul{\CType(S)}
  }
&&
\end{flalign*}
\begin{align*}
\anonrule{put}{%
  \sem{\G\vctx \ell\c\Ref_S}=\alpha\c\ul{\G}\to\sem{\Ref_S}\quad\sem{\G\vctx v\c\CType(S)}=\beta\c\ul{\G}\to\ul{\CType(S)}
  }{%
  \sem{\G\cctx \ell\ass v\c 1}=\geoU_\star(\Hid\oname{get})^H\comp\eta\comp\brks{\alpha,\beta} \c\ul{\G}\to T1 
  }
\end{align*}
\begin{align*}
&&
\anonrule{new}{%
\begin{array}{r@{}l}
  \sem{\G, x\c\Ref_{S}& \vctx v\c\CType(S)} = \alpha\c\ul{\G}\times\ul{\Ref_{S}}\to\ul{\CType(S)}\\
  \sem{\G, x\c\Ref_{S}& \cctx p\c A} = \beta\c\ul{\G}\times\ul{\Ref_{S}}\to T\ul{A}
\end{array}
  }{%
  \sem{\G\cctx\letref{x\ass v}{p}\c A} = \beta^\klstar\comp\tau\comp\Brks{\id,
\geoU_\star(\oname{new}^\klstar)^H\comp\eta\comp\curry\alpha}\c\ul{\G}\to T\ul{A}
}
&&
\end{align*}
}}
\end{center}
\caption{Semantics of non-standard language constructs.}
\label{fig:prog.sem}
\end{figure}

\subsection{Semantics of the Program Language}

Let us sketch the semantics for the language in Fig.~\ref{fig:prog.lang}. Since our language 
is a proper extension of the generic fine-grain call-by-value~\cite{LevyPowerEtAl02},
which has a standard semantics w.r.t.\ a given strong monad, we only focus on the 
non-standard term constructs. Recall that a value judgement $\G\vctx v\c A$ is 
interpreted as a morphism $\sem{\G\vctx v\c A}\c\ul{\G}\to\ul{A}$ and a
computation judgement $\G\cctx p\c A$ is interpreted as a Kleisli 
morphism $\sem{\G\cctx p\c A}\c\ul{\G}\to T\ul{A}$ where the semantics 
of types is as in Example~\ref{ex:main}, which yields the following explicit
expression for the heap functor
\begin{align*}
	H w = \prod_{(\ell\c S)\in w} \ul{\CType(S)}w
\end{align*}
The operations for working with the store are then interpreted according to the 
assignments in Fig.~\ref{fig:prog.sem}: $\eta$ and~$\tau$ refer to unit and strength 
of $T$ correspondingly; $(\argument)^\klstar$ in the rule for $\oname{letref}$ 
refers to Kleisli liftings (both of $T$ and $P$). The auxiliary natural transformations 
$\oname{put}$, $\oname{get}$ and $\oname{new}$ for writing, reading and allocating
correspondingly are defined as follows:
\begin{flalign*}
\qquad&\oname{get}\c\geoU^\star\sem{\Ref_S}\times H\to\geoU^\star\ul{\CType(S)}{\times H}\\
&\oname{get}_w(\ell\in w^\mone(S),\eta\in Hw)= (\eta(\ell),\eta)\\[1.5ex]
&\oname{put}\c (\geoU^\star\sem{\Ref_S}\times \geoU^\star\ul{\CType(S)})\times H\to 1\times H\\
&\oname{put}_w(\ell\in w^\mone(S),v\in\ul{\CType(S)}w, \eta\in Hw) = (\star,\delta)\\
&\qquad\text{where~} \delta(\ell) = v, \delta(\ell') = \eta(\ell')\text{~~if~~} \ell\nequiv \ell'\\[1.5ex]
&\oname{new}\c\geoU^\star(\ul{\CType(S)}^{\ul{\Ref_{S}}})\times H\to P(\geoU^\star\ul{\Ref_{S}}\times H)\\
&\oname{new}_w(\alpha\c\BW(w,\argument)\times\geoU^\star\ul{\Ref_{S}}\to\geoU^\star\ul{\CType(S)},\eta\in Hw)\\
&\qquad = (\inl \c w \to w\oplus\{\ell\c S\}, ((\ul{\Ref_{S}} \inr)(\ell),\\
&\qquad\qquad \Hplb(\id,\inl)(\eta)\oplus[\ell\c S\mto\alpha_{w\oplus\{\ell\c S\}}(\inl,(\ul{\Ref_{S}} \inr)(\ell))]))_\sim
\end{flalign*}
For simplicity, we only presented the case of one variable in $\oname{letref}$ -- the 
case of many variables is completely analogous.

\subsection{Proof of Theorem~\ref{thm:int-BA}}
For an internal poset $B$ in a topos $\BT$ with a subobject classifier $\Omega$, 
let $\dar\c B\to\Omega^B$ be the \emph{principal ideal} operator obtained by currying the greater-or-equal relation
$B\times B\to\Omega$. Then the internal join $\bigior\c\Omega^B \to B$ is defined 
as left order-adjoint to~$\dar$. 
In the internal language of $\BT$, $\dar x = \{y\mid y\leq x\}$ and the adjointness 
condition for $\bigior$ can be spelled out as follows:
\begin{align*}
			\forall S \in\PSet B, y \in S.& \, y \leq \bigior S\\
			\forall S \in\PSet B, x \in B.& \,
		(\forall y \in S. \, y \leq x)\impl \bigior S \leq x
\end{align*}
where $\PSet = \Omega^{(\argument)}$ is the (covariant) powerobject functor.

We will need the following connection between indexed joins~$\bigor$ as in 
Definition~\ref{def:ic-HA} and internal joins~$\bigior$.
\begin{proposition}\label{prop:int-ext}
Let $B$ be an internal poset in a topos. Then indexed join and internal 
join structures on $B$ are equivalent under the following mutual conversions:
\begin{flalign*}
&&\bigor_{f\c I\to J} (g \c I \to B) &\;= \bigl(J\xto{\curry(\oname{eq}\comp(\id\times f))} \Omega^I\xto{\exists_g}\Omega^B\xto{\bigior} \Omega\bigr) &\\
&&\bigior &\;= \bigor_{n\c \oname{\ni}_B\ito\Omega^B} (e\c \oname{\ni}_B\ito B)
\end{flalign*}
where $\oname{eq}\c J\times J\to\Omega$ is the internal equality on $J$, 
$\brks{n,e}\c\oname{\ni}_B\ito\Omega^B\times B$ is a pullback of $\top\c 1\to\Omega$
along the evaluation morphism $\ev\c\Omega^B\times B\to\Omega$. Meets are connected 
analogously.
\end{proposition}
We proceed with the proof of Theorem~\ref{thm:int-BA}.
In every De Morgan topos, in particular in $[\hat\BE,\Set]$, $2$ is a complete Boolean 
algebra under the structure induced by a retraction between $2$
and $\Omega$~(\cite[Proposition~2.6.2]{Johnstone02}). More concretely, 
$\Omega=1+\Omega_{\top}$, $\bot=\inl\c 1\to\Omega$ and the obvious lattice operations 
on~$2$ are respected by the injection $[\bot,\top]\c 2\to\Omega$. By generalities,
$2$ is an internally complete partial order, and hence an internally complete Boolean
algebra. Concretely, internal joins are defined as $\Omega^2\xto{\snd}\Omega=1+\Omega_{\top}\xto{\bot+\top\comp\bang} 2$.

In $[\hat\BE,\Set]$, the subobject 
classifier~$\Omega$ sends $w$ to the set of \emph{cosieves} 
over~$w$, i.e.\ sets of initializations $\eps\c w\into w'$, closed under composition
with arbitrary initializations $w'\into w''$. The corresponding internal joins for 
$2$ can be explicitly described as follows: 
\begin{align*}
\bigior\nolimits_w (\phi_\bot\in\Omega_w,\phi_\bot\in\Omega_w) = 
  \left\{\begin{array}{lr}
    \bot &\text{~~if~~}\phi_\top=\emptyset\\
    \top &\text{~~otherwise}
  \end{array}\right.
\end{align*}
Now, in a topos, internal 
joins are equivalent to indexed joins, which for every $f\c I\to J$, yields 
$\bigor_f\c [\hat\BE,\Set](I,2)\to [\hat\BE,\Set](J,2)$ given as follows:
$(\bigor_{f}\beta)_w(j\in Jw) = \top$ iff there exist $\eps \c w \into w'$ and $i \in I w'$ such that
$f_{w'} (i) = (J \eps) (j)$ and $\beta_{w'}(i) = \top$. In particular, note that
$\bigl(\bigor_{\nabla}[\fst,\snd]\bigr)_w(a\in 2,b\in 2) = \top$ iff $a=\top$ or~$b=\top$. 

Now, the joins on $\hat\geoU_\star(2^X)$ are obtained as compositions 
\begin{align*}
[\BW,\Set](I,\hat\geoU_\star(2^X))
\;\iso
&\;[\hat\BE,\Set](\hat\geoU^\star I\times X,2)
\xto{\bigor_{\hat\geoU^\star f\times X}} 
[\hat\BE,\Set](\hat\geoU^\star J\times X,2)\\*
\iso 
&\;[\BW,\Set](J,\hat\geoU_\star(2^X))
\end{align*}
from a general result~\cite[B2.3.7]{Johnstone02}, using the fact that 
$\hat\geoU^\star (\argument)\times X\dashv\hat\geoU^\star(\argument)^X$ where 
the left adjoint is pullback-preserving. Explicitly:
\begin{align*}
\Bigl(\bigor_f \phi\c I& \to\hat\geoU^\star(2^X)\Bigr)_w (j \in J w)(\rho\c w\to w')\\*
 =&\, \{ x \in Xw'\mid
\exists\eps\c w	'\into w'', \exists i \in I w''.\,\\*
&\qquad f_{w''}(i) = J(\hat\geoU\eps\comp \rho)(j)\land (X\eps)(x)\in 
\phi_{w''}(i)(\id_{w''})\},
\end{align*}
which is shown as follows. Note the isomorphism $\Psi\c [\BW,\Set](I,\hat\geoU_\star(2^X))
\iso[\hat\BE,\Set](\hat\geoU^\star I\times X,2)$:
\begin{align*}
	(\Psi (\alpha\c I \to\hat\geoU_\star(2^X)))_w(i,x) = \top &\;\iff 
	x\in (\alpha_w(i))(\id_w),\\
	y\in (\Psi^\mone(\beta\c\hat\geoU_\star I \times X \to 2))_w (i \in Iw)(\rho\c
	w \to w') &\;\iff\beta_{w'}((I \rho)(i), y) = \top
\end{align*}
where $x\in Xw$, $y\in Xw'$ and $i \in Iw$.
Then, given $x\in X w'$, $j\in Jw$, $\rho\c w\to w'$, and $f\c I\to J$, 
\begin{flalign*}
&&x\in\Bigl(\bigor_f&\phi\c I\to\hat\geoU^\star(2^X)\Bigr)_w (j)(\rho)\\*
&&\iff\;&x\in\Bigl(\Psi^\mone \Bigl(\bigor_{\hat\geoU^\star f\times X}\Psi\phi\Bigr)\Bigr)_w (j)(\rho)\\
&&\iff\;& \Bigl(\bigor_{\hat\geoU^\star f\times X}\Psi\phi\Bigr)_{w'} ((J\rho) j,x)=\top&\by{def.~$\Psi$}\\
&&\iff\;& \exists\eps\c w'\into w'', \exists i\in Iw'', \exists k\in Xw''.\\
&&&\quad(\hat\geoU^\star f)_{w''}(i) = ((\hat\geoU^\star J)\eps)((J\rho) j)\land\\
&&&\quad k = (X\eps)(x)\land (\Psi\phi)_{w''} (i,k)=\top&\by{def.~$\bigor_{\hat\geoU^\star f\times X}$}\\
&&\iff\;& \exists\eps\c w'\into w'', \exists i\in Iw''.\\
&&&\quad f_{w''}(i) = J(\hat\geoU\eps)((J\rho) j)\land (X\eps)(x)\in\phi_{w''} (i)(\id_{w''})&\by{def.~$\Psi$}\\
&&\iff\;& \exists\eps\c w'\into w'', \exists i\in Iw''.\\*
&&&\quad f_{w''}(i) = J(\hat\geoU\eps\comp\rho)(j)\land (X\eps)(x)\in\phi_{w''} (i)(\id_{w''}).
\end{flalign*}
By using the isomorphism~\eqref{eq:Phi-iso}, 
this yields 
the requisite explicit definition for joins. The case of meets is analogous.

Let us verify that the induced binary joins $\lor=\bigor_{\nabla}[\fst,\snd]$ are indeed
computed as pointwise set unions. Let $\phi,\psi\subseteq\hat P\hat H w$, $\rho\c w\to w'$
and $x \in Xw'$. Then by definition, 
$(\rho\comma x)_\sim\in\phi\lor_w\psi$ iff there exist $\eps\c w '\into w''$
and $\phi',\psi'\subseteq\hat P\hat H w''$, such that $(\phi',\psi') = 
((\check\PSet\comp (\hat PX) (\hat\geoU\eps\comp\rho))\phi,(\check\PSet\comp (\hat PX) (\hat\geoU\eps\comp \rho))\psi)$
and $(\id_{w''}, (X\eps)(x))_\sim\in\phi'$ or $(\id_{w''}, (X\eps)(x))_\sim\in\psi'$.
Equivalently, $(\rho\comma x)_\sim\in\phi\lor_w\psi$ iff there exist $\eps\c w '\into w''$ such that
$(\id_{w''}, (X\eps)(x))_\sim\in\oname{hide}_{\hat\geoU\eps\comp\rho}^\mone[\phi]$ or 
$(\id_{w''}, (X\eps)(x))_\sim\in\oname{hide}_{\hat\geoU\eps\comp\rho}^\mone[\psi]$.
Now,
\begin{flalign*}
&&(\id_{w''}, (X\eps)(x))_\sim&\,\in\oname{hide}_{\hat\geoU\eps\comp\rho}^\mone[\phi]\\
&&\iff&\; \oname{hide}_{\hat\geoU\eps\comp\rho} (\id_{w''}, (X\eps)(x))_\sim\in\phi\\
&&\iff&\; (\hat\geoU\eps\comp\rho,(X\eps)(x))_\sim\in\phi&\by{\eqref{eq:P-cov-act}}\\*
&&\iff&\; (\rho, x)_\sim\in\phi&\by{defn.~of~$\sim$}
\end{flalign*}
and analogously for $\psi$. In summary $(\rho\comma x)_\sim\in\phi\lor_w\psi$ iff 
$(\rho, x)_\sim\in\phi$ or $(\rho, x)_\sim\in\psi$, as required. 
Analogously, binary meets are computed as set intersections, by a dual argument.
\qed

\subsection{Proof of Proposition~\ref{prop:p-exp}}
Let $X\c\hat\BE\to\Set$ and $w\in |\hat\BE|$. Then
\begin{align*}
(\check\PSet\comp P X) w 
\;\iso&\;([\BW,2^{(\argument)}]\comp P X) w \\
=&\; \Set\biggl(\int^{\rho\c w\to w'\in w \dar \hat\geoU} X w',2\biggr)\\*
\iso&\; \int_{\rho\c w\to w'\in w \dar \hat\geoU} \Set(X w',2)\\[1ex]
\iso&\; (\hat\geoU_\star 2^X)w
\end{align*}
which is functorial in $w$ and natural in $X$.\qed

\subsection{Semantics of the Logic: Classical Case}
Let us first spell out the isomorphism
\begin{align}\label{eq:pred-iso}
(\check\PSet\comp \hat P\hat H)^X 
\iso (\hat\geoU_\star(2^{\hat H}))^X
\iso \hat\geoU_\star(2^{\hat\geoU^\star X\times\hat H})
\iso \check\PSet\comp \hat P(\hat\geoU^\star X\times\hat H).
\end{align}
Recall that by definition, $(Y^X)w = \Nat(\BW(w,\argument)\times X,Y)$
for any $X,Y\c\BW\to\Set$. Then $\alpha\c\BW(w,\argument)\times X\to\check\PSet\comp (\hat P\hat H)$
corresponds to $\alpha'\c\BW(w,\argument)\times X\to\hat\geoU_\star(2^{\hat H})$,
such that 
\begin{align*}
&(\rho'\c w'\to w'',\eta\in\hat Hw'')_\sim\in\alpha(\rho\c w\to w', x\in Xw')\\
\text{iff}\qquad &(\eta\in\hat Hw'')\in\alpha'(\rho\c w\to w', x\in Xw')(\rho'\c w'\to w'').
\intertext{Then $\alpha'$ corresponds to $\phi'\in \hat\geoU_\star(2^{\hat\geoU^\star X\times\hat H}) w$,
such that}
&(\eta\in\hat Hw'')\in\alpha'(\rho\c w\to w', x\in Xw')(\rho'\c w'\to w'')\\
\text{iff}\qquad &((X\rho') x\in Xw'',\eta\in\hat Hw'')\in\phi'(\rho'\comp\rho\c w\to w'').
\intertext{
Finally, $\phi'$ corresponds to $\phi\in\check\PSet\comp\hat P(\hat\geoU^\star X\times\hat H)$
such that}
&(x\in Xw',\eta\in\hat H w')\in\phi'(\rho\c w\to w')\\
\text{iff}\qquad &(\rho\c w\to w',(x\in Xw',\eta\in\hat H w'))_\sim\in\phi.
\intertext{
In summary,~\eqref{eq:pred-iso} connects $\alpha$ with $\phi$ in such a way that
} 
&(\rho'\c w'\to w'',\eta\in\hat Hw'')_\sim\in\alpha(\rho\c w\to w', x\in Xw')\\
\text{iff}\qquad &(\rho'\comp\rho\c w\to w'',((X\rho')x\in Xw'',\eta\in\hat H w''))_\sim\in\phi.
\end{align*}
For the backward implication of this equivalence one can always assume $\rho'=\id_{w'}$ 
-- the fact that any other choice of $\rho'$ produces the same result follows 
from naturality of $\alpha$.

Given $w\in |\BW|$, $s\in\ul{\G}w$, $\rho\c w\to w'$, $\eta\in\hat Hw'$, $a\in Xw'$, let us use  
$s,\rho,\eta\models\phi$ as a synonym for $(\rho,\eta)_\sim\in\sem{\Gamma\ctx\phi\c\pro}_w(s)$
and $s,\rho,(a,\eta)\models\phi$ as a synonym for $(\rho,(a,\eta))_\sim\in\sem{\Gamma\ctx\phi\c\pred A}_w(s)$.
Unfolding the abstract constructions yields the following semantics of terms 
by induction:
\begin{itemize}[itemsep=1ex]
\item $s,\rho,\eta\models\bot$ ~~holds never;
\item $s,\rho,\eta\models\top$ ~~holds always;
\item $s,\rho,\eta\models\phi\tto\psi$ ~~if~~ $s,\rho,\eta\models\phi$ implies $s,\rho,\eta\models\psi$;
\item $s,\rho,\eta\models\phi\land\psi$ ~~if~~ $s,\rho,\eta\models\phi$ and $s,\rho,\eta\models\psi$;
\item $s,\rho,\eta\models\phi\lor\psi$ ~~if~~ $s,\rho,\eta\models\phi$ or $s,\rho,\eta\models\psi$;
\item $s,\rho,\eta\models\phi(v)$ ~~if~~ $s,\rho,((\sem{\G\vctx v\c A}_{w'}\comp\ul\Gamma\rho)s,\eta)\models\phi$; 
\item $s,\rho,(a,\eta)\models x.\,\phi$ ~~if~~ $a=(X\rho) b$ and $(s,b),\rho,\eta\models\phi$;
\item $s,\rho,\eta\models  \ell\ito v$ ~~if~~  $\eta=(w''\subseteq w',\delta\in\CH(w'',w'))$ 
and\\[1ex]\erule\quad $\delta(r\c S) = (\sem{\G\vctx v\c \CType(S)}_{w'}\comp\ul{\G}\rho)s$  
\\[1ex]\erule\quad where $(\sem{\G\vctx \ell\c\Ref_S}_{w'}\comp\ul{\G}\rho)s = (r\c S)\in w''$;
\item $s,\rho,\eta\models v = u$ ~~if~~
  $(\sem{\G\vctx v\c A}_{w''}\comp\ul{\G}\rho'\comp\ul{\G}\rho)(s) = (\sem{\G\vctx u\c A}_{w''}\comp\ul{\G}\rho'\comp\ul{\G}\rho)(s)$
  for some $\rho'\c w'\to w''$;
\item $s,\rho,\eta\models\exists\phi$ ~~if~~ $\ul{\G}(\hat\geoU\epsilon\comp\rho)s,\id_{w''},(a,\hat H\epsilon\comp\eta)\models\phi$ for some 
$\epsilon\c w'\into w''$, $a\in\ul{A} w''$;
\item $s,\rho,\eta\models\forall\phi$ ~~if~~ $\ul{\G}(\hat\geoU\epsilon\comp\rho)s,\id_{w''},(a,\hat H\epsilon\comp\eta)\models\phi$ for all 
$\epsilon\c w'\into w''$, $a\in\ul{A} w''$.
\end{itemize}

\subsection{Proof of Lemma~\ref{lem:B-cHA}}
Let us check that $\ucl_w$ is natural in $w$. Let $\rho\c w\to w'$ and let 
$\phi\subseteq\hat P \hat H w$. We need to show that 
\begin{align*}
\{(\rho',\eta')_\sim\in&\,\hat P\hat Hw' \mid \oname{hide}_\rho(\rho',\eta')_\sim\in\ucl_w(\phi)\}\\
=&\, \{(\rho',\eta')_\sim\in\hat P\hat Hw' \mid \oname{hide}_\rho(\rho',\eta)_\sim\in \phi, \eta\leq\eta'\}.
\end{align*}
Since $\oname{hide}_\rho(\rho',\eta')_\sim = (\rho'\comp\rho,\eta')_\sim$, 
$\oname{hide}_\rho(\rho',\eta')_\sim\in\ucl_w(\phi)$ iff for a suitable $\eta\leq\eta'$
$(\rho'\comp\rho,\eta)_\sim\in\phi$. The latter is the same as $\oname{hide}_\rho(\rho',\eta)_\sim\in \phi$.
Thus we obtained the desired equality.
 
By definition, $\ucl\comp\iota=\id$, i.e.\ $(\iota,\ucl)$ is an internal 
retraction. On the other hand,~$\iota$ is a right order-adjoint of $\ucl$: indeed, 
both $\iota$ and $\ucl$ are clearly monotone, hence $\phi\leq\iota(\psi)$ entails
$\ucl(\phi)\leq\ucl(\iota(\psi)) = \psi$, and $\ucl(\phi)\leq\psi$ entails 
$\phi\leq\iota(\ucl(\phi))\leq\iota(\psi)$, for, obviously, by definition, 
$\phi\leq\iota(\ucl(\phi))$.

Since upward closure is preserved by finite meets and joins, $\TA$ is an 
internal sublattice of $\check\PSet\comp\hat P\hat H$. Moreover, we transfer large
internal joins from $\check\PSet\comp\hat P\hat H$ to~$\TA$ as follows:
\begin{align}\label{eq:join-TA}
\bigior\Psi  = \ucl\Bigl(\bigior\{\iota(\psi)\mid\psi\in\Psi\}\Bigr)
\end{align}
Let us check that this definition is valid, i.e.\ that the defined $\bigior$ is 
a left order-adjoint to the principal ideal operator $\dar$:
\begin{align*}
\ucl\Bigl(\bigior\{\iota(\psi)&\mid\psi\in\Psi\}\Bigr) \leq \phi\\
\iff&\bigior\{\iota(\psi)\mid\psi\in\Psi\} \leq \iota(\phi)\\
\iff&\{\iota(\psi)\mid\psi\in\Psi\} \subseteq \dar\iota(\phi) = \{\phi'\mid \phi'\leq \iota(\phi)\}\\
\iff&\Psi=\{\psi\mid\psi\in\Psi\} \subseteq  \{\phi'\mid \phi'\leq \phi\} = \dar\phi.
\end{align*}
Thus, $\TA$ is a complete join-semilattice, and therefore, a complete lattice with 
large meets in a standard way defined as follows:
\begin{align*}
\bigiand \Psi = \bigior \{ \phi \mid \forall \psi \in \Psi. \, \phi \leq \psi \}.
\end{align*}
We would like to show that, analogously to the $\bigior$ case:
\begin{align}\label{eq:meet-TA}
\bigiand\Psi = \ucl\Bigl(\bigiand\{\iota(\psi)\mid\psi\in\Psi\}\Bigr).
\end{align}
Indeed, we have
\begin{align*}
\bigiand \Psi 
=&\; \bigior \{ \phi \mid \forall \psi \in \Psi. \, \phi \leq \psi \}\\
=&\; \ucl\Bigl(\bigior \{ \iota(\phi) \mid \forall \psi \in \Psi. \, \phi \leq \psi \}\Bigr)\\
=&\; \ucl\Bigl(\bigior \{ \phi \mid \forall \psi \in \Psi. \, \ucl(\phi) \leq \psi \}\Bigr)\\
=&\; \ucl\Bigl(\bigior \{ \phi \mid \forall \psi \in \Psi. \, \phi \leq \iota(\psi) \}\Bigr)\\
=&\; \ucl\Bigl(\bigior \{ \iota(\phi) \mid \forall \psi \in \{\iota(\psi)\mid\psi\in\Psi\}. \, \iota(\phi) \leq \psi \}\Bigr)\\
=&\; \ucl\Bigl(\bigiand\{\iota(\psi)\mid\psi\in\Psi\}\Bigr).
\end{align*}

\noindent Distributivity of binary meets over infinite joints is transported from 
$\PSet\circ (\hat P \hat H)$ to $\TA$ as follows:
\begin{align*}
\phi\land\bigior\Psi 
=&\;\ucl\Bigl(\iota(\phi)\land\iota\Bigl(\ucl\Bigl(\bigior\{\iota(\psi)\mid\psi\in\Psi\}\Bigr)\Bigr)\Bigr)\\
=&\;\ucl\Bigl(\iota(\phi)\land\bigior\{\iota(\psi)\mid\psi\in\Psi\}\Bigr)\\
=&\;\ucl\Bigl(\bigior\{\iota(\phi)\land\iota(\psi)\mid\psi\in\Psi\}\Bigr)\\
=&\;\ucl\Bigl(\bigior\{\iota(\ucl(\iota(\phi)\land\iota(\psi)))\mid\psi\in\Psi\}\Bigr)\\
=&\;\bigior\{\phi\land\psi\mid\psi\in\Psi\}.
\end{align*}
Hence, $\TA$ is an internally complete Heyting algebra.

From~\eqref{eq:join-TA} and~\eqref{eq:meet-TA}, we obtain the corresponding 
formulas for indexed meets and joins:
\begin{align*}
\bigor_f \phi = \ucl \circ \bigor_f (\iota \circ \phi)
&&
\bigand_f \phi = \ucl \circ \bigand_f (\iota \circ \phi)
\end{align*}
where $f\c I\to J$ and $\phi\c I \to \TA$. Indeed, e.g.\ for joins, using 
Proposition~\ref{prop:int-ext}:
\begin{align*}
	\bigor_f (\phi \c I \to \TA) (j\in Jw) 
	=&\; \bigior \{ \phi (i)\mid f(i) = j \}\\
	=&\; \ucl \Bigl(\bigior \{ \iota (\phi (i))\mid f (i) = j \}\Bigr)\\
	=&\; \Bigl(\ucl \circ \bigor_f (\iota \circ \phi)\Bigr) (j),
\end{align*}
and analogously for meets. Using Theorem~\ref{thm:int-BA}, we obtain 
explicit formulas for $\bigor$ and $\bigand$ on $\TA$:
\begin{align*}
\Bigl(\bigor_f\phi\c I& \to\TA\Bigr)_w (j \in J w)\\*
 =&\, \ucl\{ (\rho\c w\to w', \eta \in\hat Hw')_\sim\mid
\exists\eps\c w	'\into w'', \exists i \in I w''.\,\\*
&\qquad f_{w''}(i) = J(\hat\geoU\eps\comp \rho)(j)\land (\id_{w''}, (\hat H\eps)(\eta))_\sim\in 
\phi_{w''}(i)\},\\
\Bigl(\bigand_f\phi\c I& \to\TA\Bigr)_w (j \in J w)\\
 =&\, \ucl\{ (\rho\c w\to w', \eta \in \hat Hw')_\sim\mid
\forall\eps\c w'\into w'', \forall i \in I w''.\,\\
&\qquad f_{w''}(i) = J(\hat\geoU\eps\comp \rho)(j)\impl (\id_{w''}, (\hat H\eps)(\eta))_\sim\in 
\phi_{w''}(i)\}.
\end{align*}
Finally, we show that the applications of $\ucl$ have no effect, as the
resulting sets are already upward closed.\mpnote{Here we use that join is a well-defined
operation, i.e.\ it does not depend on the chosen representative of the equivalence
class. I've proven that.}
Assume that $\eps\c w	'\into w''$, $i \in I w''$, $f_{w''}(i) = J(\hat\geoU\eps\comp \rho)(j)$,
$(\id_{w''}, (\hat H\eps)(\eta))_\sim\in\phi_{w''}(i)$ and $\eta\leq\eta'$.
Since $\phi_{w''}(i)$ is upward closed and $\hat H \epsilon$
is order preserving, also $(\id_{w''}, (\hat H \hat \epsilon)\eta')_\sim \in \phi_{w''} (i)$. 
The analogous argument applies to meets. Now, $\TA$ is an internal complete sublattice 
of $\check\PSet\comp\hat P\hat H$, since meets and joins are computed using the same 
formulas.
\qed

\subsection{Proof of Lemma~\ref{lem:cl-preserve}}
By definition, $\phi \star_w \psi \subseteq \ucl_w (\phi \star_w \psi)$. We show the
converse. Assume that $(\rho\c w \to w'\comma \delta')_\sim\in\ucl_w (\phi \star_w
\psi)$. Then there exists $(\rho, (w_a \subseteq w'\comma \eta \in\Hplb(w_a, w')))_\sim\in\phi\star_w\psi$ such
that $(w_a \subseteq w', \eta \in\Hplb(w_a, w'))\leq \delta'$. By Lemma~\ref{lem:sep-props}~(iii), 
there exist $w_1$, $w_2$, such that $w_1 \uplus w_2 = w_a$, $(\rho, (w_1 \subseteq w', \Hplb(w_1
\subseteq w_a, w')\eta)_\sim\in\phi$ and $(\rho, (w_2 \subseteq w'\comma
\Hplb(w_2 \subseteq w_a, w')\eta)_\sim\in\psi$. 

By definition, $\delta'= (w_b \subseteq w', \eta'\in\Hplb(w_b, w'))$, for suitable 
$w_b$ and $\eta'$. Again, by definition, $ w_b$ can be presented as $w_a 
\cup (w_b \smin w_a) = w_1 \cup w_2 
\cup (w_b \smin w_a)$. 
By upward closure of $\psi$, $(\rho, (w_2 \cup (w_b
\smin w_a)\subseteq w'\comma \Hplb(w_2 \cup (w_b \smin w_a)\subseteq
w_b\comma w')\eta)_\sim\in\psi$ and we already know that $(\rho, (w_1 \subseteq
w', \Hplb(w_1\subseteq w_a, w')\eta))_\sim\in\phi$. Therefore,
$(\rho, \delta')_\sim = (\rho, (w_b \subseteq w', \eta'))_\sim\in (\phi \star_w \psi)$.
\qed

\end{document}